\documentclass[11pt,a4paper]{article}
\usepackage[centering,includeheadfoot,margin=2 cm]{geometry}

\usepackage[utf8]{inputenc}
\usepackage{amssymb}
\usepackage{amsmath}
\usepackage{amsfonts,epsfig}
\usepackage{color,colortbl}
\usepackage{multirow,amsthm}

\newtheorem{theorem}{\bf Theorem}[section]

\newtheorem{corollary}{\bf Corollary}[section]
\newtheorem{example}{\bf Example}[section]
\numberwithin{equation}{section}

\usepackage{epstopdf}

\newcommand{\be}{\begin{equation}}
\newcommand{\ee}{\end{equation}}

\definecolor{jobcolor}{cmyk}{1,.10,0,.10}

\begin{document}

\title{Macroscopic description of a heavy particle immersed \\ within a flow of light particles}

\author{Radek Erban\thanks{Mathematical Institute, University of Oxford, Radcliffe Observatory Quarter, Woodstock Road, Oxford, OX2 6GG, United Kingdom;
e-mail: erban@maths.ox.ac.uk} 
\hskip 1cm and \hskip 1cm Robert A. Van Gorder\thanks{Department of Mathematics and Statistics, University of Otago, P.O. Box 56, Dunedin 9054, New Zealand; \hfill\break
e-mail: robert.vangorder@otago.ac.nz}}

\maketitle

\begin{abstract}
\noindent
A micro-hydrodynamics model based on elastic collisions of light point solvent particles with a heavy solute particle is investigated in the setting where the light particles have velocity distribution corresponding to a background flow. Considering a range of stationary background flows and distributions for the solvent particle velocities, the macroscopic Langevin-type description of the behaviour of the heavy particle is derived in the form of a generalized Ornstein-Uhlenbeck process. At leading order, the drift term in this process depends upon both the geometric structure of the background flow and the size of the heavy particle, while both drift and diffusion terms scale with moments of the light particle velocity distribution. Computational methods for simulating the micro-hydrodynamics model are then designed to confirm the theoretical results. To enable long-time stochastic simulations, simulations are performed in a frame co-moving with the heavy particle. Efficient methods for sampling the position and velocity distributions of incoming solvent particles at the boundaries of the co-moving frame are derived for a range of distributions of solvent particles. The simulations show good agreement with the theoretical results.
\end{abstract}

\section{Introduction}

The motion of a heavy particle within a heat bath of light particles serves as a canonical model for the molecular theory of Brownian motion~\cite{einstein1905motion, chandrasekhar1943stochastic, mazur1970molecular}. In particular, collisions with light heat bath particles result in momentum transfer to the heavy particle, thereby inducing Brownian motion of the heavy particle. The Brownian dynamics of the heavy particle can be modelled via a Langevin equation~\cite{langevin1908theorie} treating the contribution from random light particle collisions through a diffusion term. 

Consider the motion of a heavy solute particle (a ball with radius $R$ and mass $M$) which is immersed within a solvent flow in domain $\Omega \subset {\mathbb R}^3$. The solvent is described as point particles that interact with the heavy solute particle when they collide with it. This configuration has been applied in a number of studies during the last 50 years, starting as one-dimensional~\cite{holley1971motion} and three-dimensional~\cite{durr1981mechanical} mechanical models of Brownian motion. We denote the position (centre) and velocity of the heavy solute particle by ${\mathbf X}$ and ${\mathbf V}$, respectively. We assume that all solvent particles have the same mass, denoted by $m$, and satisfying $m \ll M$. The positions and velocities of solvent particles are denoted by ${\mathbf x}_j$ and ${\mathbf v}_j$, respectively, where the particle index $j=1,2,3,\dots$ goes over all integers for theoretical studies~\cite{durr1981mechanical,Erban:2020:SMR} with spatial domain $\Omega = {\mathbb R}^3$. In contrast, there are a finite number of solvent particles in computational studies using finite sized domains. To approximate the theoretical case of $\Omega = {\mathbb R}^3$ in computational studies, the solvent particles are introduced through suitable boundary conditions~\cite{Erban:2014:MDB,gunaratne2019multi}.

Denoting the dimensionless parameter $\mu=M/m$ and assuming that the collisions of the heavy particle with point solvent particles are without friction, the conservation of momentum and energy 
yield the following formulae for post-collision velocities~\cite{durr1981mechanical,Erban:2020:SMR}:
\begin{eqnarray}\label{elcol}
\widetilde{\mathbf V}
&=& 
\left[ {\mathbf V} \right]^\parallel
+ 
\frac{\mu - 1}{\mu + 1} \, \left[ {\mathbf V} \right]^\perp
 + 
\frac{2}{\mu + 1} \, \left[{\mathbf v}_j \right]^\perp,
\label{elcol1}
\\
\widetilde{\mathbf v}_j 
&=& 
\left[ {\mathbf v}_j \right]^\parallel
+
\frac{1 - \mu}{\mu + 1} \, \left[{\mathbf v}_j \right]^\perp
 + 
\frac{2 \mu}{\mu + 1} \, \left[ {\mathbf V} \right]^\perp\,.
\label{elcol2}
\end{eqnarray}
Here, ${\mathbf v}_j$ is the velocity of the solvent particle which collided with the heavy solute particle, tildes denote post-collision velocities, superscripts $\perp$ denote projections of velocities on the line through the centre of the solute particle and the 
collision point on its surface, and superscripts $\parallel$ denote tangential components. Micro-hydrodynamics models based on elastic collisions~(\ref{elcol1})--(\ref{elcol2}) have been studied by a number of authors~\cite{holley1971motion, durr1981mechanical, dunkel2006relativistic, Erban:2014:MDB, Erban:2020:SMR, gunaratne2019multi}, all of whom investigate the macroscopic behaviour of the heavy particle for $\mu \gg 1$, where it provides a mechanical description of Brownian motion. 

For example, consider the case of infinite domain $\Omega = \mathbb{R}^3$ containing an infinite number of solvent particles with positions distributed according to a spatial Poisson process with constant density
\begin{equation}
\lambda_{\mu} =
\frac{3}{8 R^2} \sqrt{\frac{(\mu + 1) \, \gamma}{2 \pi D}},
\label{lambda3Dexp}
\end{equation}
where $D$ is a diffusion coefficient having
dimension $[D] = [\mbox{length}]^2 \, [\mbox{time}]^{-1}$ 
and $\gamma$ is a friction coefficient having 
dimension $[\gamma] = [\mbox{time}]^{-1}$. Let the velocities of the solvent particles be distributed according to the Gaussian (Maxwell-Boltzmann) distribution 
\begin{equation}
f_\mu({\mathbf v})
= \frac{1}{\sigma_\mu^3(2\pi)^{3/2}} 
\exp \left[ 
- \frac{v_1^2+v_2^2+v_3^2}{2 \sigma_\mu^2} \right],
\label{fvel3Dexp}
\end{equation}
where ${\mathbf v} = (v_1,v_2,v_3)^{\mathrm{T}}$ and
\begin{equation}
\sigma_\mu = \sqrt{(\mu + 1) \, D \, \gamma } \, ,
\label{sigmaold}
\end{equation}
and assume that a solvent particle moves (between collisions) using the free flight, i.e. the position and velocity of the $j$-th solvent particle satisfy
\begin{equation}
\frac{\mbox{d}{\mathbf x}_j}{\mbox{d}t} 
=
{\mathbf v}_j \, 
\qquad
\mbox{and}
\qquad
\frac{\mbox{d}{\mathbf v}_j}{\mbox{d}t} 
=
{\mathbf 0} 
\label{freeflight}
\end{equation}
between the elastic collisions governed by~(\ref{elcol1})--(\ref{elcol2}). Then, it can be shown~\cite{durr1981mechanical,Erban:2014:MDB} that, in the limit $\mu \to \infty$, the behaviour of the heavy particle converges (in the sense of distributions) 
to the Langevin dynamics given by
\begin{eqnarray}
\mbox{d}{\mathbf X} 
& 
{\hskip -3mm} = {\hskip -3mm}
& 
{\mathbf V} \; \mbox{d}t, 
\label{langeq}
\\
\mbox{d}{\mathbf V} 
& 
{\hskip -3mm} = {\hskip -3mm}
&
- 
\gamma {\mathbf V}\, \mbox{d}t 
+ 
\gamma \sqrt{2 D} \; \mbox{d}{\mathbf W}, 
\label{langeqV1}
\end{eqnarray}
where ${\mathbf W}$ is a three-dimensional vector of independent Wiener processes.

Aside from standard Brownian motion of a heavy particle within a heat bath as governed by~(\ref{langeq})--(\ref{langeqV1}), there has been an interest in manipulating or controlling the motion of a heavy particle in a heat bath in some manner~\cite{wulfert2017driven, goswami2019work}, with temperature gradients~\cite{nicolis1965evaluation, perez1994brownian} and electromagnetic fields~\cite{pal2014extracting, tothova2020brownian} used to bias the motion of the heavy particle. There is a history of literature extending the Brownian motion of a heavy object within a fluid flow comprising lighter particles~\cite{saffman1976brownian, batchelor1977effect, ramshaw1979brownian, russel1981brownian, miyazaki1995brownian}. In many of these studies macroscale models are formulated for the heavy particle motion, informed by heat bath statistics of the type mentioned above. Compare this with literature on the flow of heavy spheres within a given flow where the interaction of the sphere and the surrounding fluid is modelled at the macroscopic scale from the onset of the problem~\cite{stokes1851effect, ho1974inertial, maxey1983equation, feng1994direct}. The Brownian motion of particles within certain flows has been suggested as a possible route to anomalous diffusion~\cite{viecelli1993statistical, solomon1993observation, klafter1996beyond}, motivating the development of more accurate models bridging microscale Brownian motion with macroscale flows~\cite{gallagher2019newton}. The development of more accurate models of Brownian motion of a heavy particle immersed within generic flows aides us in better understanding the theory behind recent experimental results on the dynamics of Brownian particles within flows~\cite{li2010measurement}.

It is possible to reconcile models of Brownian motion of a heavy particle in a heat bath with macroscopic models of a heavy particle within a fluid flow by generalizing the heat bath to account for the directed motion of the heat bath particles according to a prescribed fluid flow. In some works, a Langevin model for the motion of a heavy particle within a background flow is asserted~\cite{rubi1988brownian, gotoh1990brownian, katayama1996brownian, garbaczewski1998diffusion, orihara2011brownian, wang2022generalized}. A derivation of Langevin dynamics for a heavy particle immersed in a background flow field taking the form of a linear shear flow was presented in~\cite{dobson2013derivation}. 

In this paper, we extend the study of the Brownian motion of a heavy particle within a heat bath to account for general flows of the light particles, as well as for general forms of the light particle velocity distribution. Indeed, the latter generalization is motivated by heavy-tailed velocity distributions that find application in high-energy granular gas experiments~\cite{rouyer2000velocity, aranson2002velocity, van1998velocity, kohlstedt2005velocity} and more generic composite velocity distributions that find application in the study of plasma flows~\cite{cairns1995electrostatic, izacard2017generalized}, to give two examples. We take a first principles approach to modelling, starting with the conservation of momentum and energy and upscaling to Langevin dynamics for generic stationary flows and velocity distributions. We consider the case where the solvent particles are immersed in stationary flow described by
\begin{equation}
{\mathbf u}(\mathbf x)
=
(u_1({\mathbf x}),
u_2({\mathbf x}),
u_3({\mathbf x}))^{\mathrm{T}}\,,
\label{statflow}
\end{equation}
where ${\mathbf x} = (x_1,x_2,x_3)^{\mathrm{T}} \in \Omega$ and the velocity distribution (\ref{fvel3Dexp}) is generalized to be both ${\mathbf u}$-dependent and non-Gaussian as $f_\mu({\mathbf v},{\mathbf u})$. Since this means that the velocity distribution $f_\mu({\mathbf v},{\mathbf u})$ is, in general, position-dependent, the free flight equations~(\ref{freeflight}) have to be generalized to include the ${\mathbf u}$-dependence as well. This is done in Section~\ref{sec2}, where we introduce our microscopic model of the solvent in detail. We then derive a Langevin description for the motion of the heavy solute particle immersed within a heat bath comprising solvent particles moving with a prescribed flow profile in Section~\ref{sec3}. The theoretical results we obtain take the form of a generalized Ornstein-Uhlenbeck process and are valid for generic forms of the velocity distribution of the solvent particles comprising the heat bath (which is of use even for heat bath configurations deviating from the standard Maxwell-Boltzmann statistics~\cite{mo2019highly}), as well as for generic incompressible stationary flows. We then develop an efficient computational algorithm -- in a frame co-moving with the heavy solute particle -- to illustrate the theoretical results in Section~\ref{sec4}. We discuss the key findings and possible future directions in Section~\ref{sec5}.

\section{Motion of the solvent particles}
\label{sec2}

We assume that the mean motion of the heat bath is described by velocity field $\mathbf{u}(\mathbf{x})$ which is a stationary solution of the Navier-Stokes equations. In particular, $\mathbf{u}(\mathbf{x})$ satisfies the continuity equation written for incompressible flow as
\begin{equation}
\label{contequation}
\nabla \cdot \mathbf{u} = 0.
\end{equation}
We assume that solvent particle velocities are sampled according to a distribution \begin{equation}\label{Fgeneral}
f_\mu({\mathbf v},{\mathbf u})
= 
F_{\mu}(\mathbf{v}-\mathbf{u})
\end{equation}
where the function $F_\mu: {\mathbb R}^3 \to [0,\infty)$ satisfies the properties 
\begin{equation}\label{Fproperties}
\int_{\mathbb{R}^3} \! F_\mu(\mathbf{q}) \, \mathrm{d}\mathbf{q} =1\, \qquad \text{and} \qquad \int_{\mathbb{R}^3} \parallel\!\mathbf{q}\!\parallel^4  \!F_\mu(\mathbf{q}) \, \mathrm{d}\mathbf{q} < \infty \,,
\end{equation}
for all $\mu \geq 0$. Equation~(\ref{Fgeneral}) is a generalization of equation~(\ref{fvel3Dexp}), which is covered by our framework for $\mathbf{u} \equiv \mathbf{0}$ and $F_\mu$ being the Gaussian distribution. In general, equation~(\ref{Fgeneral}) states that the velocity distribution is centered around the flow solution $\mathbf{u}$ to equation \eqref{contequation} and is therefore $\mathbf{u}$-dependent. Some important cases of the function $F_\mu: {\mathbb R}^3 \to [0,\infty)$ will be given in the product form
\begin{equation}
\label{productform}
F_\mu (\mathbf{q})
= \frac{1}{\sigma_\mu^3} 
\,
{\mathcal F}\!\left(\!\frac{q_1}{\sigma_\mu} \!\right) 
\,
{\mathcal F}\!\left(\!\frac{q_2}{\sigma_\mu} \!\right)
\,
{\mathcal F}\!\left(\!\frac{q_3}{\sigma_\mu} \!\right),
\end{equation}
where function ${\mathcal F}: {\mathbb R} \to [0,\infty)$ is given as the Gaussian, Laplace and generalized Gaussian distributions, respectively, see Section~\ref{sec3}, but our initial derivation will consider $F_\mu$ in its full generality satisfying conditions~(\ref{Fproperties}). 

We assume that the particle velocities are distributed according to distribution~(\ref{Fgeneral}) at any time~$t \ge 0$ in our theoretical investigations. This assumption implies that the position and velocity of the $j$-th solvent particle evolve according to the equations
\begin{equation}
\dfrac{\mathrm{d}\mathbf{x}_j}{\mathrm{d}t} ={\mathbf v}_j  \, 
\qquad
\mbox{and}
\qquad
\dfrac{\mathrm{d}\mathbf{v}_j}{\mathrm{d}t} =
\left(\nabla {\mathbf u} \right){\mathbf v}_j
\equiv \sum_{\ell=1}^3 \frac{\partial {\mathbf u}}{\partial x_\ell} \, v_{j,\ell} \, ,
\label{xveqgen}
\end{equation}%
where ${\mathbf v}_j = (v_{j,1},v_{j,2},v_{j,3})^{\mathrm{T}}$. This is a generalization of 
free-flight equations~(\ref{freeflight}), which we obtain in the special case of zero flow~$\mathbf{u}(\mathbf{x}) \equiv \mathbf{0}$ in equations~(\ref{xveqgen}). 

In our illustrative simulations in Section~\ref{sec4}, we initialize the velocities of particles according to distribution~(\ref{Fgeneral}) at time $t=0$ and let their positions and velocities evolve according to equations~(\ref{xveqgen}). Assuming that the $j$-th particle initial velocity is sampled as $\mathbf{v}_j^0$, its velocity at time $t$ to satisfies
$$
\mathbf{v}_j(t) = \mathbf{u}(\mathbf{x}_j(t))+\mathbf{v}_j^0.
$$
This relation implies that the velocity distribution~(\ref{Fgeneral}) is preserved for all time $t>0$ and we have (for $\Delta t > 0$) 
$$
\mathbf{v}_j(t+\Delta t) - \mathbf{v}_j(t) 
= 
\mathbf{u}\big(\mathbf{x}_j(t+\Delta t)\big) - \mathbf{u}(\mathbf{x}_j(t)) 
= 
\mathbf{u}\big(\mathbf{x}_j(t)+\mathbf{v}_j(t)\Delta t + \mathcal{O}( \Delta t)\big) - \mathbf{u}(\mathbf{x}_j(t)).
$$
Applying the Taylor expansion, we get
\begin{equation*}
\dfrac{\mathrm{d}\mathbf{v}_j}{\mathrm{d}t} 
= 
\lim_{\Delta t \rightarrow 0}\dfrac{\mathbf{v}_j(t+\Delta t) - \mathbf{v}_j(t)}{\Delta t} 
= 
\lim_{\Delta t \rightarrow 0} \left\lbrace \left( \nabla \mathbf{u}\right)\mathbf{v}_j(t)  + \mathcal{O}\left( \Delta t\right) \right\rbrace 
=  \left( \nabla \mathbf{u}\right)\mathbf{v}_j\, ,
\end{equation*}
which is equation~(\ref{xveqgen}).

\section{Derivation of Langevin dynamics}
\label{sec3}

The generalization of the Langevin dynamics~(\ref{langeq})--(\ref{langeqV1}) for the motion of a heavy solute particle immersed within a heat bath comprising solvent particles moving with a prescribed flow profile~${\mathbf u}$ can be written as 
\begin{eqnarray}
\mbox{d}{\mathbf X} 
& 
{\hskip -3mm} = {\hskip -3mm}
& 
{\mathbf V} \; \mbox{d}t, 
\label{SDEmainX}
\\
\mbox{d}{\mathbf V} 
& 
{\hskip -3mm} = {\hskip -3mm}
&
\boldsymbol{\alpha}(\mathbf{X},\mathbf{V}) \, \mathrm{d}t 
+ 
\boldsymbol{\beta}(\mathbf{X},\mathbf{V}) \; \mbox{d}{\mathbf W}, 
\label{SDEmainY}
\end{eqnarray}
where ${\mathbf W}$ is a three-dimensional vector of independent Wiener processes and the coefficients $\boldsymbol{\alpha}(\mathbf{X},\mathbf{V})$ 
and 
$\boldsymbol{\beta}(\mathbf{X},\mathbf{V})$ depend on the velocity field $\mathbf{u}(\mathbf{x})$. To calculate these coefficients, we first express them as integrals over the surface of the heavy particle
\begin{equation}
\mathcal{S}(\mathbf{X},R)
\equiv
\mathcal{S}(\mathbf{X}(t),R) 
= 
\left\lbrace \mathbf{y}\in \mathbb{R}^3 \; \Big| \; \parallel\!\mathbf{y}-\mathbf{X}(t)\! \parallel \, = R \right\rbrace.
\label{partsurface}
\end{equation}
This is done, for general function $F_\mu$ satisfying propertises (\ref{Fproperties}), in the following theorem.

\begin{theorem}\label{theorem1}
Let $\mathbf{y}$ be a point on the surface~(\ref{partsurface}) of the heavy particle, i.e., $\mathbf{y}\in\mathcal{S}(\mathbf{X},R)$, and let vectors $\boldsymbol{\eta}_2 \in \mathbb{R}^3$ and $\boldsymbol{\eta}_3 \in \mathbb{R}^3$ be chosen so that $\left\lbrace (\mathbf{y}-\mathbf{X})/R, \boldsymbol{\eta}_2, \boldsymbol{\eta}_3 \right\rbrace$ comprise an orthonormal basis for $\mathbb{R}^3$. Define the functions
\begin{align}
&\psi_\ell(\mathbf{y}) = - \dfrac{2\lambda_\mu \left( y_\ell - X_\ell(t)\right)}{(1+\mu)R}\int_0^\infty \int_{-\infty}^\infty \int_{-\infty}^\infty \left\lbrace \xi_1 + \mathbf{V}(t)\cdot \left( \dfrac{\mathbf{y}-\mathbf{X}(t)}{R}\right)\right\rbrace^2 \nonumber \\
& \quad \times F_\mu \left(  - \left(\xi_1 + \mathbf{u}\cdot \frac{\mathbf{y}-\mathbf{X}(t)}{R}\right) \left( \dfrac{\mathbf{y}-\mathbf{X}(t)}{R}\right)
+
\sum_{i=2}^3 
\left(\xi_i - \mathbf{u}\cdot \boldsymbol{\eta}_i  \right) \boldsymbol{\eta}_i  \right) \mathrm{d}\xi_3 \, \mathrm{d}\xi_2 \, \mathrm{d}\xi_1 \label{thmproof6}
\end{align}
and 
\begin{align}
&\phi_{\ell,j}(\mathbf{y})  = \dfrac{4\lambda_\mu \left( y_\ell - X_\ell(t)\right)\left( y_j - X_j(t)\right)}{(1+\mu)^2R^2}\int_0^\infty \int_{-\infty}^\infty \int_{-\infty}^\infty \left\lbrace \xi_1 + \mathbf{V}(t)\cdot \left( \dfrac{\mathbf{y}-\mathbf{X}(t)}{R}\right)\right\rbrace^3 \nonumber\\
& \quad \times F_\mu \left( - \left(\xi_1 + \mathbf{u}\cdot \frac{\mathbf{y}-\mathbf{X}(t)}{R}\right) \left( \dfrac{\mathbf{y}-\mathbf{X}(t)}{R}\right) 
+ 
\sum_{i=2}^3 
\left(\xi_i - \mathbf{u}\cdot \boldsymbol{\eta}_i  \right) \boldsymbol{\eta}_i \right) \mathrm{d}\xi_3 \,\mathrm{d}\xi_2 \,\mathrm{d}\xi_1 \,. \label{thmproof7}
\end{align}
Then the drift coefficient of the It\^{o} stochastic differential equations (\ref{SDEmainX})--(\ref{SDEmainY}) for the motion of the heavy solute particle can be expressed as the surface integral
\begin{equation}\label{alphageneral}
\alpha_\ell(\mathbf{X},\mathbf{V}) = \int_{\mathcal{S}(\mathbf{X}(t),R)} \psi_\ell(\mathbf{y})\,\mathrm{d}A\,,\quad \ell =1,2,3\,,
\end{equation}
where $\mathrm{d}A$ is the surface element  centred at $\mathbf{y}$. 
The diffusion terms are given by the square root of the matrix having entries
\begin{equation}
 \beta_{\ell,j}^2(\mathbf{X},\mathbf{V}) =  \int_{\mathcal{S}(\mathbf{X}(t),R)} \phi_{\ell,j}(\mathbf{y})\,\mathrm{d}A\,, \quad \ell,j =1,2,3\,.
 \label{betageneral}
\end{equation}
\end{theorem}
\begin{proof}
We assume, in the limit of large $\mu$, that the velocity of the heavy particle evolves according to discretized SDE~(\ref{SDEmainY}) which can be written as
\begin{equation}\label{thmproof1}
V_\ell(t+\Delta t) = V_\ell(t) + \alpha_\ell(t) \Delta t + \beta_\ell(t) \sqrt{\Delta t} \, \chi_\ell \,,
\end{equation}
where $\Delta t$ is a (small) time step and $\chi_\ell$, for $\ell=1,2,3,$ is a normally distributed random number with zero mean and unit variance. To determine $\boldsymbol{\alpha}$ and $\boldsymbol{\beta}$ we will match the mean and variance of the velocity jump in the microscopic solvent model to those of \eqref{thmproof1}.

Let $\mathbf{y}\in\mathcal{S}(\mathbf{X},R)$ be a point on the surface of the heavy particle. We first find the distribution measuring the average change in velocity of the heavy particle due to collisions near the surface point $\mathbf{y}$ during the time interval $(t,t+\Delta t)$; we will show that this quantity is to the leading order in $\Delta t$ given as $\psi_\ell(\mathbf{y}) \, \Delta t$, where $\psi_\ell(\mathbf{y})$ is expressed by~(\ref{thmproof6}). Then, $\psi_\ell(\mathbf{y}) \, \Delta t \, \mathrm{d}A$ is the average change of the $\ell$-th component of the velocity of the heavy molecule caused by collisions with heat bath particles in the time interval $(t,t+\Delta t)$ on the surface element $\mathrm{d}A$ centred at $\mathbf{y}$. 

Consider a heat bath particle located at point $\mathbf{x}$ at time $t$ which collides with the heavy particle at time $t+\tau \in (t,t+\Delta t)$ at the surface point which had position $\mathbf{y}$ at time $t$. Since $\Delta t$ is small and the velocity jump of the heavy particle due to collision with any one heat bath particle is small, we approximate $\mathbf{V}$ to be a constant in the interval $(t,t+\Delta t)$. That is to say, the change in~$\mathbf{V}$ is $\mathcal{O}(\sqrt{\Delta t})$. At this level of approximation, the coordinate of the surface point at the collision time~$t+\tau$ is equal to $\mathbf{y} + \tau \mathbf{V}$. Since the heat bath molecule moved from $\mathbf{x}$ to the collision point $\mathbf{y}+\tau \mathbf{V}$, its velocity before the collision was
$$
\mathbf{v}=\dfrac{\mathbf{y}+\tau \mathbf{V}(t)-\mathbf{x}}{\tau} = \mathbf{V}(t) + \dfrac{\mathbf{y}-\mathbf{x}}{\tau} \,.
$$
Making use of \eqref{elcol}, we write the change in velocity of the heavy particle due to the collision as
\begin{equation}\label{thmproof2}
\widetilde{\mathbf{V}}-\mathbf{V} = \dfrac{2}{1+\mu} \left[ \mathbf{v} - \mathbf{V}(t) \right]^\perp = \dfrac{2}{1+\mu}\left\lbrace \left( \mathbf{v}-\mathbf{V}(t)\right)\cdot \left( \dfrac{\mathbf{y}-\mathbf{X}(t)}{R}\right)\right\rbrace \left( \dfrac{\mathbf{y}-\mathbf{X}(t)}{R}\right) .
\end{equation}
The position $\mathbf{x}$ of the heat bath molecule must be in the half space above the plane tangent to the heavy particle at the collision point $\mathbf{y}+\tau \mathbf{V}(t)$; in particular, this means that the component of the velocity $\mathbf{v}$ in the direction $\mathbf{y}-\mathbf{X}(t)$ must be negative. We then parameterize the allowable velocities by 
\begin{equation}\label{v_param}
\mathbf{v} = - \xi_1 \left( \dfrac{\mathbf{y}-\mathbf{X}(t)}{R}\right) + \xi_2 \boldsymbol{\eta}_2 + \xi_3 \boldsymbol{\eta}_3\,,
\end{equation}
where $\xi_1 >0$, $\xi_2, \xi_3 \in \mathbb{R}$, and $(\mathbf{y}-\mathbf{X}(t))/R,$ $\boldsymbol{\eta}_2$ and $\boldsymbol{\eta}_3$ comprise an orthonormal basis for $\mathbb{R}^3$. Using this basis to represent velocity vectors, we rewrite \eqref{thmproof2} as
\begin{equation}\label{thmproof3}
\widetilde{\mathbf{V}}-\mathbf{V} = -\dfrac{2}{(1+\mu)R}\left\lbrace \xi_1 + \mathbf{V}(t)\cdot \left( \dfrac{\mathbf{y}-\mathbf{X}(t)}{R}\right)\right\rbrace \left( \mathbf{y}-\mathbf{X}(t)\right)\,.
\end{equation}
For a given heat bath particle velocity $\mathbf{v}$, the set of possible starting points allowing for a collision with the heavy particle within the surface area element $\mathrm{d}A$ centred at $\mathbf{y}$ during the time interval $(t,t+\Delta t)$ is a cylinder of cross-sectional area $\mathrm{d}A$ and perpendicular height
$$
h = \Delta t \left( \mathbf{V}(t) - \mathbf{v}\right)\cdot \left( \dfrac{\mathbf{y}-\mathbf{X}(t)}{R}\right) = \Delta t \left\lbrace \xi_1 + \mathbf{V}(t)\cdot \left( \dfrac{\mathbf{y}-\mathbf{X}(t)}{R}\right)\right\rbrace \,.
$$
We assume the number of heat bath particles in this cylinder is Poisson distributed with mean $\lambda_\mu$ times its volume $h \, \mathrm{d}A$. The probability of collision of a heat bath particle with velocity in the interval $(\mathbf{v},\mathbf{v}+\mathrm{d}\mathbf{v})$ with the surface element $\mathrm{d}A$ over the time interval $(t,t+\Delta t)$ reads
\begin{equation}\label{thmproof4}
\lambda_\mu \Delta t \left\lbrace \xi_1 + \mathbf{V}(t)\cdot \left( \dfrac{\mathbf{y}-\mathbf{X}(t)}{R}\right)\right\rbrace \mathrm{d}A \, f_\mu \left( - \xi_1 \left( \dfrac{\mathbf{y}-\mathbf{X}(t)}{R}\right) + \xi_2 \boldsymbol{\eta}_2 + \xi_3 \boldsymbol{\eta}_3 \right)\mathrm{d}\mathbf{v}\,.
\end{equation}
To find the average change in velocity of the heavy particle due to collisions with $\mathrm{d}A$ during the time interval $(t,t+\Delta t)$, we multiply~\eqref{thmproof3} by \eqref{thmproof4} and integrate over all possible velocities $\mathbf{v}$ parameterized by \eqref{v_param}, obtaining
\begin{align*}
\psi_\ell(\mathbf{y}) \,
\Delta t \, \mathrm{d}A & 
= - \Delta t \, 
\mathrm{d}A \,
\dfrac{2\lambda_\mu \left( y_\ell - X_\ell(t)\right)}{(1+\mu)R}\int_0^\infty \int_{-\infty}^\infty \int_{-\infty}^\infty \left\lbrace \xi_1 + \mathbf{V}(t)\cdot \left( \dfrac{\mathbf{y}-\mathbf{X}(t)}{R}\right)\right\rbrace^2 \nonumber \\
& \qquad \times f_\mu \left( - \xi_1 \left( \dfrac{\mathbf{y}-\mathbf{X}(t)}{R}\right) + \xi_2 \boldsymbol{\eta}_2 + \xi_3 \boldsymbol{\eta}_3 \right) \mathrm{d}\xi_3 \, \mathrm{d}\xi_2 \, \mathrm{d}\xi_1 \nonumber \\
& = - \Delta t \, \mathrm{d}A \, \dfrac{2\lambda_\mu \left( y_\ell - X_\ell(t)\right)}{(1+\mu)R}\int_0^\infty \int_{-\infty}^\infty \int_{-\infty}^\infty \left\lbrace \xi_1 + \mathbf{V}(t)\cdot \left( \dfrac{\mathbf{y}-\mathbf{X}(t)}{R}\right)\right\rbrace^2 \nonumber \\
& \quad \times F_\mu \left(  - \left(\xi_1 + \mathbf{u}\cdot \frac{\mathbf{y}-\mathbf{X}(t)}{R}\right) \left( \dfrac{\mathbf{y}-\mathbf{X}(t)}{R}\right)  
+ \sum_{i=2}^3
\left(\xi_i - \mathbf{u}\cdot \boldsymbol{\eta}_i  \right) \boldsymbol{\eta}_i \right) \mathrm{d}\xi_3 \, \mathrm{d}\xi_2 \, \mathrm{d}\xi_1 \,. 
\end{align*}
Note that we have expressed $\mathbf{u}$ in terms of the basis vectors $\left\lbrace
(\mathbf{y}-\mathbf{X}(t))/R, \boldsymbol{\eta}_2, \boldsymbol{\eta}_3 \right\rbrace$, obtaining the representation 
$$
\mathbf{u} = \left\lbrace \mathbf{u}\cdot \frac{\mathbf{y}-\mathbf{X}(t)}{R}\right\rbrace \left( \frac{\mathbf{y}-\mathbf{X}(t)}{R} \right) + \left( \mathbf{u}\cdot \boldsymbol{\eta}_2 \right) \boldsymbol{\eta}_2 + \left( \mathbf{u}\cdot \boldsymbol{\eta}_3 \right) \boldsymbol{\eta}_3
.$$ From here, we integrate this expression over the surface of the heavy molecule, and upon equating this integral to the average velocity jump in \eqref{thmproof1}, we obtain
\eqref{alphageneral} in the limit $\Delta t \to 0$.  

To determine the diffusion term, we calculate the variance in the velocity jump in the $\ell$-th direction from time $t$ to time $t+\Delta t$ and equate this with $\beta_\ell^2(t) \Delta t$ in equation~(\ref{thmproof1}). It is also possible to have off diagonal terms in the diffusion tensor, and we account for these as well. Since the mean velocity jump is $\mathcal{O}(\Delta t)$, the variance (to leading order) is the second moment of the velocity jump. We square the quantity \eqref{thmproof3}, multiply it by \eqref{thmproof4}, and then integrate over all possible parameterizations of $\mathbf{v}$ in \eqref{v_param}, finding that the variance in the $\ell$-th component of the velocity jump of the heavy particle due to collisions with $\mathrm{d}A$ over the time interval $(t,t+\Delta t)$ is $\phi_{\ell,j}(\mathbf{y}) \,
\Delta t \, \mathrm{d}A$ where
\begin{align*}
\phi_{\ell,j}(\mathbf{y}) & = \dfrac{4\lambda_\mu \left( y_\ell - X_\ell(t)\right)\left( y_j - X_j(t)\right)}{(1+\mu)^2R^2}\int_0^\infty \int_{-\infty}^\infty \int_{-\infty}^\infty \left\lbrace \xi_1 + \mathbf{V}(t)\cdot \left( \dfrac{\mathbf{y}-\mathbf{X}(t)}{R}\right)\right\rbrace^3 \nonumber\\
& \qquad \times f_\mu \left( - \xi_1 \left( \dfrac{\mathbf{y}-\mathbf{X}(t)}{R}\right) + \xi_2 \boldsymbol{\eta}_2 + \xi_3 \boldsymbol{\eta}_3 \right) \mathrm{d}\xi_3 \,\mathrm{d}\xi_2 \, \mathrm{d}\xi_1 \nonumber \\
& = \dfrac{4\lambda_\mu \left( y_\ell - X_\ell(t)\right)\left( y_j - X_j(t)\right)}{(1+\mu)^2R^2}\int_0^\infty \int_{-\infty}^\infty \int_{-\infty}^\infty \left\lbrace \xi_1 + \mathbf{V}(t)\cdot \left( \dfrac{\mathbf{y}-\mathbf{X}(t)}{R}\right)\right\rbrace^3 \nonumber\\
& \quad \times F_\mu \left( - \left(\xi_1 + \mathbf{u}\cdot \frac{\mathbf{y}-\mathbf{X}(t)}{R}\right) \left( \dfrac{\mathbf{y}-\mathbf{X}(t)}{R}\right) +
\sum_{i=2}^3
\left(\xi_i - \mathbf{u}\cdot \boldsymbol{\eta}_i  \right) \boldsymbol{\eta}_i \right) \mathrm{d}\xi_3 \, \mathrm{d}\xi_2 \, \mathrm{d}\xi_1 \,.
\end{align*}
Integrating $\phi_{\ell,j}(\mathbf{y}) \,
\Delta t \, \mathrm{d}A$ over the surface $\mathcal{S}(\mathbf{X}(t),R)$ of the heavy molecule, and matching with second moment terms  $\beta_{\ell,j}^2(t) \Delta t$ in equation~(\ref{thmproof1}), we obtain 
$3\times 3$ matrix $\boldsymbol{\beta}^2$ for which we then need to find the matrix root. However, since $\beta_{\ell,j}^2(t)=\beta_{j,\ell}^2(t)$ by symmetry of equation~(\ref{thmproof7}), the matrix~$\boldsymbol{\beta}^2$ is real symmetric and hence such a root exists. 
\end{proof}

\noindent
Theorem \ref{theorem1} has given us a procedure by which to determine the drift and diffusion terms once a distribution $F_\mu$ is known. While this procedure is general, it does not convey much insight. For this reason, we will consider distributions for which we can obtain a better analytical understanding of the drift and diffusion terms. As we will see, this will still involve a rather generic class of distributions, holding physically relevant examples -- such as the Gaussian distribution~(\ref{fvel3Dexp}) -- as special cases.

\subsection{Generalizations of the Gaussian distribution}

To make further progress, note that Theorem \ref{theorem1} has involved finding a triple integral over the moments of certain generic probability density functions. There is always one direction of integration along which a moment is contributed, and two which are effectively inert. Define the \textit{marginal} density function
\begin{equation}\label{mdfIntegral}
\widehat{f}_\mu(\xi_1) = \int_{-\infty}^\infty \int_{-\infty}^\infty f_\mu \!\left(\! - \xi_1 \left( \dfrac{\mathbf{y}-\mathbf{X}(t)}{R}\right) + \xi_2 \boldsymbol{\eta}_2 + \xi_3 \boldsymbol{\eta}_3 \right)\mathrm{d}\xi_3 \, \mathrm{d}\xi_2 \,,
\end{equation}
where we have integrated over the two inert directions. Assuming that this integration can be carried out, let us assume the result takes the functional form 
\begin{equation}\label{mdfForm}
\widehat{f}_\mu(\xi) = \dfrac{1}{\sigma_\mu} \,\mathcal{F}\!\left( \!\dfrac{\xi}{\sigma_\mu}\!\right).
\end{equation}
Here $\mathcal{F}$ is still kept fairly general, and is assumed to inherit properties of the density function $F_\mu$ in~\eqref{Fproperties}. We also assume the marginal density function \eqref{mdfForm} satisfies the moment relations
\begin{equation}\label{mdfxnScale}
\int_{a}^\infty \xi^\kappa \, \widehat{f}_\mu(\xi) \, \mathrm{d}\xi = \mathcal{F}_\kappa \, \sigma_\mu^\kappa + \mathcal{O}\!\left( \!\dfrac{1}{\sigma_\mu}\!\right)
\quad
\mbox{as} \;\;
\sigma_\mu \rightarrow \infty,
\quad
\mbox{for} \;\;
\kappa = 0,1,2,3
\end{equation}
where the leading term is independent of the specific finite value of $a$. The asymptotic relation~(\ref{mdfxnScale}) follows for many distributions from 
$$
\int_{a}^\infty \xi^\kappa \, \widehat{f}_\mu(\xi) \, \mathrm{d}\xi 
= 
\int_{a}^\infty \xi^\kappa \dfrac{1}{\sigma_\mu} \, \mathcal{F}\!\left( \!\dfrac{\xi}{\sigma_\mu}\!\right) \mathrm{d}\xi = 
\sigma_\mu^\kappa 
\int_{a/\sigma_\mu}^\infty q^\kappa \, \mathcal{F}\!\left( q\right) \mathrm{d}q
=
\mathcal{F}_\kappa \, \sigma_\mu^\kappa
-
\sigma_\mu^\kappa 
\int_0^{a/\sigma_\mu} q^\kappa \, \mathcal{F}\!\left( q\right) \mathrm{d}q
$$
where
\begin{equation}
\label{fkformula}    
\mathcal{F}_\kappa
=
\int_{0}^\infty q^\kappa \, \mathcal{F}\!\left( q\right) \mathrm{d}q.
\end{equation}
A list of leading-order terms in~\eqref{mdfxnScale} is provided for a variety of specific distributions in Table~\ref{tab:mdfxnscalinglist}. 

\begin{table}
\caption{{\it Asymptotic scale constants in equation~$(\ref{mdfxnScale})$ for various marginal density functions $\mathcal{F}$, given by~$(\ref{fkformula})$.} \label{tab:mdfxnscalinglist}}
\begin{center}
\begin{tabular}{lccccc}
\rowcolor{jobcolor} \color{white}{Distribution} & \color{white}{$\mathcal{F}(q)$} & \color{white}{$\mathcal{F}_0$} & \color{white}{$\mathcal{F}_1$} & \color{white}{$\mathcal{F}_2$} & \color{white}{$\mathcal{F}_3$}\\
Gaussian  & \raise -4.1mm \hbox{\rule{0pt}{10.5mm}}$\dfrac{1}{\sqrt{2\pi}} \, \exp\!\left( \!-\dfrac{q^2}{2}\right)$  & {\hskip -2mm}
$\dfrac{1}{2}$ {\hskip -2mm} & $\dfrac{1}{\sqrt{2\pi}}$ &  $\dfrac{1}{2}$  &  $\sqrt{\dfrac{2}{\pi}}$ \\
      \hline
Laplace & \raise -4.1mm \hbox{\rule{0pt}{10.2mm}}$\dfrac{1}{2\sqrt{2}}\, \exp\!\left(\! -\dfrac{|q|}{\sqrt{2}}\right)$  & {\hskip -2mm} $\dfrac{1}{2}$ {\hskip -2mm} & $\dfrac{1}{\sqrt{2}}$ &  $2$  &  $6\sqrt{2}$ \\
\hline
\hskip -1mm \raise -2mm \hbox{\hsize=3.8cm\vbox{generalized Gaussian \hfill ($\theta=3$ in (\ref{gengauss})) \hfill}} {\hskip -3mm} & \raise -5.4mm \hbox{\rule{0pt}{12.6mm}}$\dfrac{3}{2 \sqrt{2}\,\Gamma(1/3)} \, \exp\!\left(\! - \left|\dfrac{q}{\sqrt{2}}\right|^{3}\right)$  {\hskip -1mm} & {\hskip -2mm} $\dfrac{1}{2}$ {\hskip -2mm} & {\hskip -2mm} $\dfrac{2^{-5/6}\Gamma(5/6)}{\sqrt{\pi}}$ {\hskip -2mm} & {\hskip -2mm}  $\dfrac{1}{\Gamma(1/3)}$  {\hskip -2mm} & {\hskip -2mm}  $\dfrac{\sqrt{2}}{3}$ \\
\hline
hyperbolic secant & \raise -4.1mm  \hbox{\rule{0pt}{10.1mm}}$\dfrac{1}{2} \, \mathrm{sech} \!\left( \dfrac{\pi q}{2}\right)$  & {\hskip -2mm} $\dfrac{1}{2}$ {\hskip -2mm} & $0.37122$ &  $\dfrac{1}{2}$  &  $0.97464$ \\
\hline
uniform &  \raise -6.4mm \hbox{\rule{0pt}{15mm}}$\begin{cases} \dfrac{1}{2} & \text{if}~|q|\leq 1 \\
0 & \text{otherwise}
\end{cases}$  & {\hskip -2mm} $\dfrac{1}{2}$ {\hskip -2mm} & $\dfrac{1}{4}$ &  $\dfrac{1}{6}$  &  $\dfrac{1}{8}$ \\
\hline
\end{tabular}
\end{center}
\end{table}

\begin{theorem}\label{theorem2}
Assume that the marginal density function $\widehat{f}_\mu(\xi_1)$ in \eqref{mdfIntegral} takes the form \eqref{mdfForm} and satisfies the asymptotic scalings \eqref{mdfxnScale}. The It\^{o} stochastic differential equation \eqref{SDEmainX}--\eqref{SDEmainY} for the motion of the heavy solute particle has drift and diffusion terms which take the form 
\begin{subequations}\label{Thm2alphabeta}
\begin{equation}
\boldsymbol{\alpha}(\mathbf{X},\mathbf{V}) = - \sqrt{2\pi} \, \mathcal{F}_1 \, \gamma \, \Big( \mathbf{V} - \mathbf{u}(\mathbf{X}) - \boldsymbol{A}_R\big[\mathbf{u}(\mathbf{y})-\mathbf{u}(\mathbf{X})\big]\Big) \,+\, \mathcal{O}\!\left(\! \dfrac{1}{\sigma_\mu}\!\right)\label{Thm2alpha}
\end{equation}
and 
\begin{equation}
\boldsymbol{\beta}(\mathbf{X},\mathbf{V}) = \sqrt{\sqrt{2\pi}\,\mathcal{F}_3 \, D} \, \gamma \, \mathbf{I} \, + \, \mathcal{O}\!\left(\! \dfrac{1}{\sigma_\mu}\!\right)
\qquad
\mbox{as}
\quad \sigma_\mu \rightarrow \infty,\label{Thm2beta}
\end{equation}
\end{subequations}
where operator $\boldsymbol{A}_R$ given by
\begin{equation} \label{A_correction}
\boldsymbol{A}_R[\mathbf{w}] = \dfrac{3}{4\pi R^2}\int_{\mathcal{S}(\mathbf{X}(t),R)} \dfrac{\mathbf{y}-\mathbf{X}(t)}{R}\left\lbrace \mathbf{w}(\mathbf{y}) \cdot \dfrac{\mathbf{y}-\mathbf{X}(t)}{R}\right\rbrace \, \mathrm{d}A
\end{equation}
incorporates the geometry of the flow and finite-size effects of the heavy particle. 
\end{theorem}
\begin{proof}
Using \eqref{mdfIntegral} in \eqref{thmproof6}, we obtain
$$
\psi_\ell(\mathbf{y}) = - \dfrac{2\lambda_\mu \left( y_\ell - X_\ell(t)\right)}{(1+\mu)R } \int_0^\infty \left\lbrace \xi_1 + \mathbf{V}(t)\cdot \left( \dfrac{\mathbf{y}-\mathbf{X}(t)}{R}\right)\right\rbrace^2  \widehat{f}_\mu \left( \xi_1 + \mathbf{u}\cdot \frac{\mathbf{y}-\mathbf{X}(t)}{R} \right) \mathrm{d}\xi_1. 
$$
Using~(\ref{lambda3Dexp}) and (\ref{sigmaold}), we get
$$
\psi_\ell(\mathbf{y}) = - \dfrac{3\gamma\left( y_\ell - X_\ell(t)\right)}{4 \sigma_\mu R^3 \sqrt{2\pi}} 
\int_{\mathbf{u}\cdot \frac{\mathbf{y}-\mathbf{X}(t)}{R}}^\infty \left\lbrace \xi + \left(\mathbf{V}(t) - \mathbf{u}\right)\cdot \left( \dfrac{\mathbf{y}-\mathbf{X}(t)}{R}\right)\right\rbrace^2  \widehat{f}_\mu(\xi) \, \mathrm{d}\xi \,.
$$
Using \eqref{mdfxnScale}, we obtain
$$
\psi_\ell(\mathbf{y}) = - \dfrac{3\gamma}{4 R^2 \sqrt{2\pi}} \left(\dfrac{y_\ell - X_\ell(t)}{R}\right)
\left( \mathcal{F}_2 \sigma_\mu + 2\mathcal{F}_1 \left(\mathbf{V}(t) - \mathbf{u}\right)\cdot \left( \dfrac{\mathbf{y}-\mathbf{X}(t)}{R}\right) + \mathcal{O}\!\left( \!\dfrac{1}{\sigma_\mu}\!\right) 
\right)
$$
as $\sigma_\mu \rightarrow \infty$. Substituting into~(\ref{alphageneral}) and using the following integration results on the surface of the sphere
\begin{equation*}
\int_{\mathcal{S}(\mathbf{X}(t),R)}\left(\dfrac{y_\ell - X_\ell(t)}{R}\right)\mathrm{d}\mathbf{y} = 0 \quad \text{for all}\quad \ell=1,2,3\,, \quad \text{and}
\end{equation*}
\begin{equation}
\int_{\mathcal{S}(\mathbf{X}(t),R)}\left(\dfrac{y_\ell - X_\ell(t)}{R}\right)\left(\dfrac{y_j - X_j(t)}{R}\right)\mathrm{d}\mathbf{y} = \dfrac{4\pi R^2}{3} \, \delta_{\ell,j} \quad \text{for all}\quad \ell,j=1,2,3 \, ,
\label{sphereformulas}
\end{equation}
then the drift term $\boldsymbol{\alpha}$ can be evaluated as
\begin{align*}
\boldsymbol{\alpha}(\mathbf{X},\mathbf{V}) & =  - \dfrac{3\mathcal{F}_1\gamma}{2 R^2 \sqrt{2\pi}}\int_{\mathcal{S}(\mathbf{X}(t),R)}\left( \dfrac{\mathbf{y}-\mathbf{X}(t)}{R}\right)\left\lbrace \left(\mathbf{V}(t) - \mathbf{u}(\mathbf{y})\right)\cdot \left( \dfrac{\mathbf{y}-\mathbf{X}(t)}{R}\right)  \right\rbrace\mathrm{d}\mathbf{y} + \mathcal{O}\!\left( \!\dfrac{1}{\sigma_\mu}\!\right) \\
& = - \dfrac{3\mathcal{F}_1\gamma}{2 R^2 \sqrt{2\pi}}\left[ \int_{\mathcal{S}(\mathbf{X}(t),R)}\left( \dfrac{\mathbf{y}-\mathbf{X}(t)}{R}\right)\left\lbrace \left(\mathbf{V}(t) -\mathbf{u}(\mathbf{X}(t))\right)\cdot \left( \dfrac{\mathbf{y}-\mathbf{X}(t)}{R}\right)  \right\rbrace\mathrm{d}\mathbf{y}
\right. \\
& \qquad \quad \left. - \int_{\mathcal{S}(\mathbf{X}(t),R)}\left( \dfrac{\mathbf{y}-\mathbf{X}(t)}{R}\right)\left\lbrace \left(\mathbf{u}(\mathbf{y}) - \mathbf{u}(\mathbf{X}(t)) \right)\cdot \left( \dfrac{\mathbf{y}-\mathbf{X}(t)}{R}\right) \right\rbrace \mathrm{d}\mathbf{y}
 \right] + \mathcal{O}\!\left(\! \dfrac{1}{\sigma_\mu}\!\right) ,
\end{align*}
as $\sigma_\mu \rightarrow \infty$, which simplifies to~(\ref{Thm2alpha}). Next, to simplify the diffusion terms, we use \eqref{mdfIntegral}, \eqref{mdfxnScale}, (\ref{lambda3Dexp}) and (\ref{sigmaold}) in \eqref{thmproof7} to obtain
\begin{align}
\phi_{\ell,j}(\mathbf{y}) & = \dfrac{4\lambda_\mu \left( y_\ell - X_\ell(t)\right)\left( y_j - X_j(t)\right)}{(1+\mu)^2R^2} \int_{\mathbf{u}\cdot \frac{\mathbf{y}-\mathbf{X}(t)}{R}}^\infty \left\lbrace \xi + \left(\mathbf{V}(t) - \mathbf{u}\right)\cdot \left( \dfrac{\mathbf{y}-\mathbf{X}(t)}{R}\right)\right\rbrace^3  \widehat{f}_\mu (\xi) \, \mathrm{d}\xi \nonumber \\
& =  \dfrac{3D\gamma^2}{2 R^2 \sqrt{2\pi}}\dfrac{\left( y_\ell - X_\ell(t)\right)\left( y_j - X_j(t)\right)}{R^2}  \left\lbrace \mathcal{F}_3 + \mathcal{O}\!\left( \!\dfrac{1}{\sigma_\mu}\! \right) \right\rbrace \nonumber\,.
\end{align}
Substituting into~(\ref{betageneral}) and using~(\ref{sphereformulas}), we get
\begin{align*}
\beta_{\ell,j}^2 & = \dfrac{3D\gamma^2}{2R^2\sqrt{2\pi}}\int_{\mathcal{S}(\mathbf{X}(t),R)} \dfrac{\left( y_\ell - X_\ell(t)\right)\left( y_j - X_j(t)\right)}{R^2}  \left\lbrace \mathcal{F}_3 + \mathcal{O}\!\left( \!\dfrac{1}{\sigma_\mu} \!\right) \right\rbrace \mathrm{d}A\\
& = \sqrt{2\pi} \, \mathcal{F}_3 \, D \, \gamma^2 \, \delta_{\ell,j} + \mathcal{O}\!\left(\! \dfrac{1}{\sigma_\mu}\! \right) \quad \text{for all}\quad \ell,j=1,2,3\,.
\end{align*}
As the squared diffusion tensor $\boldsymbol{\beta}^2$ is diagonal up to leading order, and since the full diffusion tensor $\boldsymbol{\beta}^2$ is real symmetric, we have that the square root matrix exists and is diagonal up to leading order, resulting in~(\ref{Thm2beta}).
\end{proof}

\begin{corollary}\label{CorMB}
If the light particles are distributed according to Maxwell-Boltzmann statistics~(\ref{fvel3Dexp}), then $\mathcal{F}$ is Gaussian, $\mathcal{F}_1 = 1/\sqrt{2\pi}$, $\mathcal{F}_3 = \sqrt{2/\pi}$ and formulas \eqref{Thm2alphabeta} reduce to
\begin{subequations}\label{MBDistnalphabeta}
\begin{align}
\boldsymbol{\alpha}(\mathbf{X},\mathbf{V}) &= - \gamma \left\lbrace \mathbf{V} - \mathbf{u}(\mathbf{X}) - \boldsymbol{A}[\mathbf{u}(\mathbf{y})-\mathbf{u}(\mathbf{X})]\right\rbrace + \mathcal{O}\!\left(\! \dfrac{1}{\sigma_\mu}\!\right)\,, \\
\boldsymbol{\beta}(\mathbf{X},\mathbf{V}) &= \sqrt{2D}\, \gamma \, \mathbf{I} + \mathcal{O}\!\left(\! \dfrac{1}{\sigma_\mu}\!\right),
\qquad \mbox{as} \quad \sigma_\mu \rightarrow \infty.
\end{align}
\end{subequations}
\end{corollary}

\noindent
Corollary~\ref{CorMB} shows that Theorem~\ref{theorem2} reduces to the result (\ref{langeqV1}) for the special case ${\mathbf u} \equiv {\mathbf 0}$. Next, we consider more general velocity distributions $F_\mu(\mathbf{q})$ with thin or heavy tails. For any fixed shape parameter $\theta >0$, a generalised Gaussian distribution takes the form 
\begin{equation}\label{genGauss1}
F_\mu(\mathbf{q};\theta) = \left(\dfrac{\theta}{2\sqrt{2}\,\Gamma(1/\theta)\,\sigma_\mu}\right)^{\! 3}\exp\!\left( \!-\left(\frac{\parallel\!\mathbf{q}\!\parallel}{\sqrt{2} \, \sigma_\mu}\right)^{\!\theta}\right).
\end{equation}
This may seem like a distribution that will not separate into the product form~(\ref{productform}) for general $\theta$, making the calculation of a marginal distribution complicated or even impossible. However, this all depends upon how we define the vector norm $\parallel\!\mathbf{q}\!\parallel$. Let us take 
$$
\parallel\!\mathbf{q}\!\parallel
=
\parallel\!\mathbf{q}\!\parallel_\theta  := \left(\sum_{\ell =1}^3 |q_\ell|^\theta \right)^{\! 1/\theta}.
$$ Then, \eqref{genGauss1} becomes
\begin{equation}\label{genGauss2}
F_\mu(\mathbf{q};\theta) 
= 
\prod_{\ell =1}^3 \dfrac{\theta}{2\sqrt{2} \,\Gamma(1/\theta)\,\sigma_\mu}\,\exp\!\left( \!-\left(\frac{|q_\ell|}{\sqrt{2}\,\sigma_\mu}\right)^{\!\theta}\right)\,.
\end{equation}
Although the $\theta=2$ (Maxwellian) distribution is most common in the statistical physics literature, we remark that other values of $\theta$ have applications. In particular, heavy-tailed distributions with $\theta=1$ (Laplace)~\cite{kohlstedt2005velocity} and $\theta = 3/2$~\cite{rouyer2000velocity, aranson2002velocity, van1998velocity} are arguably better distributions for certain granular gases in experiments with high-energy particles. Considering more detailed all-atom molecular dynamics models of solvent~\cite{mdbrownian,dna1,Shin:2010:BMM}, non-Gaussian distributions of forces can be estimated from simulations and used to parameterize coarse-grained Brownian dynamics and Langevin dynamics models~\cite{cmdnongaussian,utterson}. 

\begin{theorem}\label{theorem3}
Let the distribution of solvent particle velocities follow a generalised Gaussian distribution of the form \eqref{genGauss1} with vector norm $\parallel\!\mathbf{q}\!\parallel=\parallel\!\mathbf{q}\!\parallel_\theta$, then drift and diffusion terms in \eqref{Thm2alphabeta} reduce to
\begin{subequations}\label{GenGaussianalphabeta}
\begin{align}
\boldsymbol{\alpha}(\mathbf{X},\mathbf{V}) & = - \chi_1(\theta) \, \gamma \left\lbrace \mathbf{V} - \mathbf{u}(\mathbf{X}) - \boldsymbol{A}[\mathbf{u}(\mathbf{y})-\mathbf{u}(\mathbf{X})]\right\rbrace + \mathcal{O}\!\left( \!\dfrac{1}{\sigma_\mu}\!\right),\\
\boldsymbol{\beta}(\mathbf{X},\mathbf{V}) & = \chi_2(\theta)\sqrt{2 D} \, \gamma \, \mathbf{I} + \mathcal{O}\!\left( \!\dfrac{1}{\sigma_\mu}\!\right),
\end{align}
\end{subequations}
as $\sigma_\mu \rightarrow \infty$, where the scale factors read
\begin{equation}\label{chiscalefactors}
\chi_1(\theta) = 2^{(2-\theta)/\theta} \, \Gamma\!\left( \dfrac{2+\theta}{2\theta}\right) \quad \text{and} \quad \chi_2(\theta) = \dfrac{8^{1/\theta}}{2\pi^{1/4}}\sqrt{\Gamma\!\left( \dfrac{2+\theta}{2\theta}\right)\Gamma\!\left( \dfrac{4+\theta}{2\theta}\right)} \,.
\end{equation}
Here $\Gamma$ denotes the gamma function, and the scale factors are well-defined for all $\theta >0$.
\end{theorem}
\begin{proof}
Using~\eqref{genGauss2}, we can integrate over the two inert coordinates in~\eqref{mdfIntegral}, constructing the marginal distribution function according to \eqref{mdfForm},
\begin{equation}
\mathcal{F}(q) 
= 
\dfrac{\theta}{2\sqrt{2}\,\Gamma(1/\theta)}\,\exp\!\left(\! -\left(\frac{|q|}{\sqrt{2}}\right)^{\!\theta}\right)
\label{gengauss}
\end{equation}
with $\theta >0$. Using~(\ref{fkformula}), we have
\begin{equation*}
\mathcal{F}_1 = \dfrac{2^{(4-3\theta)/(2\theta)}}{\sqrt{\pi}}
\,  \Gamma\!\left( \dfrac{2+\theta}{2\theta}\right) \quad \text{and} \quad \mathcal{F}_3 = \dfrac{8^{(4-\theta)/(2\theta)}}{\pi} \, \Gamma\!\left( \dfrac{2+\theta}{2\theta}\right)\Gamma\!\left( \dfrac{4+\theta}{2\theta}\right).
\end{equation*}
The result then follows from an application of Theorem \ref{theorem2}.
\end{proof}

\begin{table}
\caption{{\it Constants scaling the drift and diffusion terms $\eqref{GenGaussianalphabeta}$ in Theorem $\ref{theorem3}$ for various marginal density functions $\mathcal{F}$ according to $\eqref{chiscalefactors}$ (left) and plots of the scale factors $\chi_1(\theta)$ and $\chi_2(\theta)$ (right).} \label{tab:GenGaussianScales}}
\begin{center}
\vspace{-0.1in}
\begin{tabular}{lcccc}
\rowcolor{jobcolor} \color{white}{Distribution} & \color{white}{$\theta$} & \color{white}{$\chi_1(\theta)$} & \color{white}{$\chi_2(\theta)$} & \cellcolor{white} \multirow{9}{*}{\includegraphics[width=0.26\textwidth]{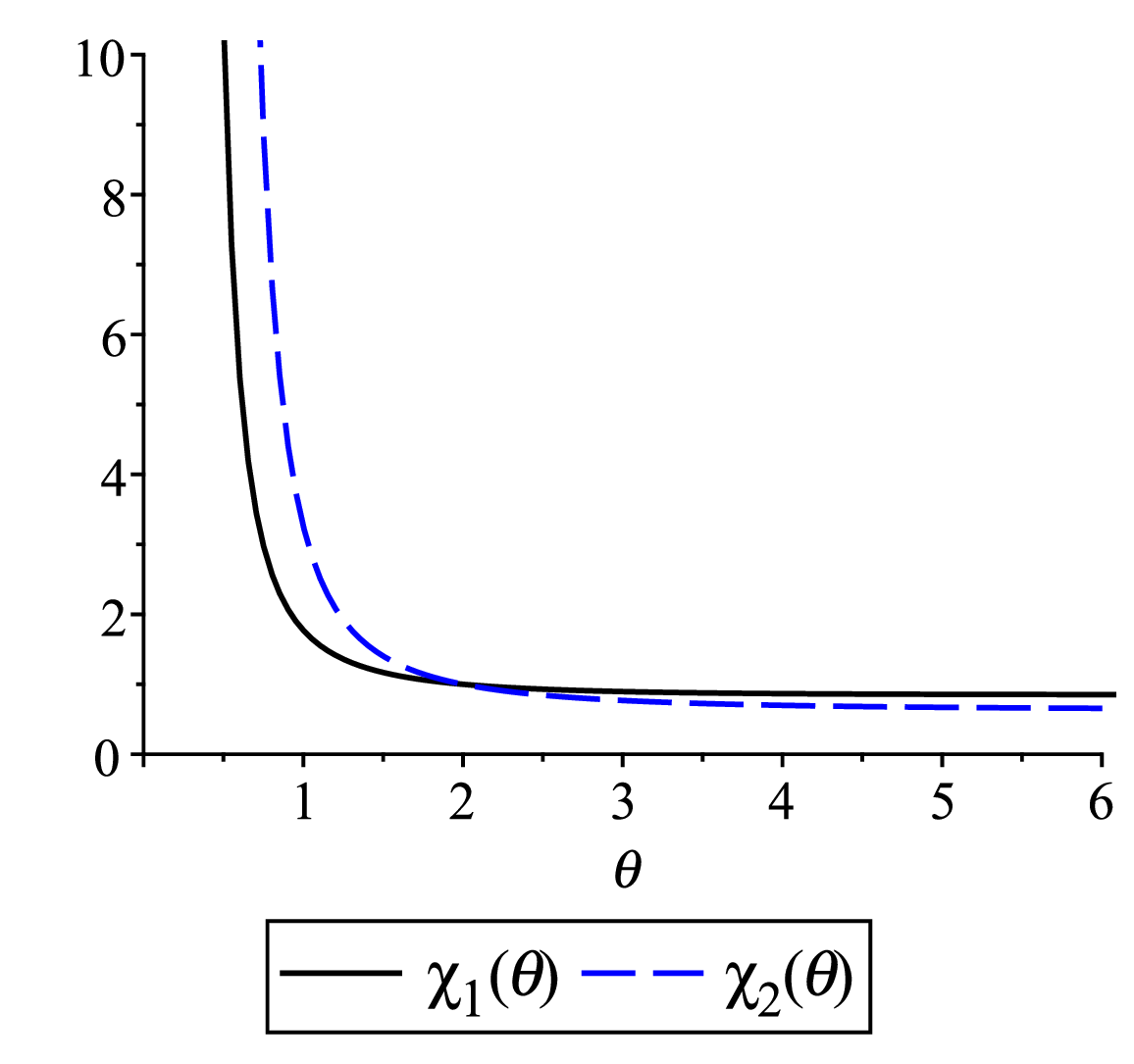}} \\
\rule{0pt}{4mm} heavy-tail limit & $\theta \rightarrow 0^+$ & $\infty$ & $\infty$ & \\
\cline{1-4}
\rule{0pt}{4mm} Laplace & $\theta = 1$ & $1.77245$ & $3.26109$  &\\
\cline{1-4}
\rule{0pt}{4mm} heavy-tailed Gaussian & $\theta = 3/2$ & $1.16885$ & $1.40335$ & \\
\cline{1-4}
\rule{0pt}{4mm} Maxwell-Boltzmann & $\theta = 2$ & $1$ & $1$ & \\
\cline{1-4}
\rule{0pt}{4mm} thin-tailed Gaussian & $\theta = 3$ & $0.89592$ & $0.76865$ & \\
\cline{1-4}
\rule{0pt}{4mm} thin-tail limit & $\theta \rightarrow \infty$ & $0.88623$ & $0.66567$ & \\
\cline{1-4} \\
& & & & \\
& & & & 
\end{tabular}
\vskip -15mm \rule{0pt}{1pt}
\end{center}
\end{table}

\noindent
We have that both scale factors~(\ref{chiscalefactors}) are decreasing functions of $\theta$, with
\begin{equation*}
\lim_{\theta \rightarrow 0^+}\chi_{1,2}(\theta) = \infty\,, \quad \lim_{\theta \rightarrow \infty}\chi_{1}(\theta) = \dfrac{\sqrt{\pi}}{2}\approx 0.88623 \,, \quad \text{and} \quad \lim_{\theta \rightarrow \infty}\chi_{2}(\theta) = \dfrac{\pi^{1/4}}{2} \approx 0.66567 \,.
\end{equation*}
Therefore, when $0< \theta < 2$, the heavy tails result in a sampling of a greater proportion of larger velocities, resulting in both faster drift and larger diffusion when compared to the Gaussian (Maxwell-Boltzmann) results. When $\theta >2$, the thin tails result in sampling of a greater proportion of smaller velocities, resulting in both slower drift and smaller diffusion when compared to the Gaussian results. We highlight key values of these scale factors in Table~\ref{tab:GenGaussianScales}, while also providing a plot of each as a function of $\theta$. We conclude this section with an example where Theorem~\ref{theorem2} is not applicable, necessitating the more general treatment of Theorem \ref{theorem1}. 

\begin{example}\label{exampleDelta}  {\rm
Assume that the solvent particle velocities are sampled according to the Dirac delta distribution $f_\mu(\mathbf{v},\mathbf{u})=F_\mu(\mathbf{v}-\mathbf{u})$, where $F_\mu(\mathbf{q)}$ is in the product form~(\ref{productform}) for $\mathcal{F}(q) = \delta(q)$. Then we have
\begin{equation*}
\int_{a}^\infty \xi^\kappa \, \widehat{f}_\mu(\xi) \, \mathrm{d}\xi = \int_{a}^\infty \xi^\kappa \dfrac{1}{\sigma_\mu} \, \delta\!\left(\! \dfrac{\xi}{\sigma_\mu}\!\right) \mathrm{d}\xi = \begin{cases}
1 - H(a) & \text{if}~~ \kappa =0\,,\\
    0 & \text{if}~~ \kappa =1,2,3\,,
\end{cases}
\end{equation*}
where $H(\cdot)$ is the Heaviside step function, so the marginal distribution depends on $a$ and yet not on $\sigma_\mu$.  Define $\mathcal{S}_- (\mathbf{X}(t),R) \subset \mathcal{S}(\mathbf{X}(t),R)$ to be the set
\begin{equation*}
\mathcal{S}_- (\mathbf{X}(t),R) = \left\lbrace \mathbf{y}\in \mathcal{S}(\mathbf{X}(t),R) \quad\text{and}\quad \mathbf{u}(\mathbf{y})\cdot \left( \dfrac{\mathbf{y}-\mathbf{X}}{R}\right)<0\right\rbrace\,.
\end{equation*}
Applying Theorem \ref{theorem1}, the drift coefficient of the It\^{o} stochastic differential equations (\ref{SDEmainX})--(\ref{SDEmainY}) for the motion of the heavy solute particle can be expressed as the surface integral
\begin{equation*}
\alpha_\ell(\mathbf{X},\mathbf{V}) = - \dfrac{2\lambda_\mu }{1+\mu} \int_{\mathcal{S}_-(\mathbf{X}(t),R)} \left(\dfrac{ y_\ell - X_\ell(t)}{R}\right)\left\lbrace \left(\mathbf{V}(t) - \mathbf{u}\right)\cdot \left( \dfrac{\mathbf{y}-\mathbf{X}(t)}{R}\right)\right\rbrace^2 \,\mathrm{d}A\,,\quad \ell =1,2,3\,,
\end{equation*}
where $\mathrm{d}A$ is the surface element centred at $\mathbf{y}$. Meanwhile, Theorem \ref{theorem1} implies that the diffusion terms are given by the square root of the matrix having entries
\begin{equation*}
 \beta_{\ell,j}^2(\mathbf{X},\mathbf{V}) = \dfrac{4\lambda_\mu}{(1+\mu)^2} \int_{\mathcal{S}_-(\mathbf{X}(t),R)} \!\!\left(\! \frac{y_\ell - X_\ell(t)}{R}\!\right)\!\left(\!\frac{y_j - X_j(t)}{R}\!\right) \left\lbrace \!\left(\mathbf{V}(t) - \mathbf{u}\right)\cdot \left( \dfrac{\mathbf{y}-\mathbf{X}(t)}{R}\right)\!\right\rbrace^{\!3} \!\mathrm{d}A
\end{equation*}
for each pair $\ell,j =1,2,3$. Since the integral is performed over $\mathcal{S}_-$ rather than $\mathcal{S}$ this builds an asymmetry into the dynamics, where the structure of this asymmetry strongly depends upon the geometry of the flow $\mathbf{u}(\mathbf{y})$ local to the heavy particle. }
\end{example}

\subsection{Velocities distributed as a sum of densities}

We consider a distribution function for the velocities of the light particles given by 
\begin{equation} \label{FSum}
f_\mu(\mathbf{v},\mathbf{u}) 
= 
F_\mu(\mathbf{v}-\mathbf{u}) \quad \text{where} \quad F_\mu(\mathbf{q})= \sum_{\ell=1}^{\ell^*}c_\ell \, F_{\mu,\ell}(\mathbf{q})
\end{equation}
and each function $F_{\mu,\ell}$ obeys the properties \eqref{Fproperties} while the coefficients are subject to the constraint $\sum_{\ell =1}^{\ell^*} c_\ell =1$, so that in turn $F_\mu(\mathbf{q})$ also obeys the properties \eqref{Fproperties}.  Probability densities viewed as a sum of simpler densities find application in the statistical understanding of certain plasma flows \cite{izacard2017generalized, cairns1995electrostatic}, where they allow for better agreement with experimental observations than would single simple distributions.
We assume that each $F_{\mu,\ell}$ has corresponding marginal distribution function $\widehat{F}_{\mu,\ell}(\xi)=\sigma_\mu^{-1} \widetilde{\mathcal{F}}_{\ell}(\xi/\sigma_\mu)$ satisfying the asymptotic scaling
\begin{equation}
\int_a^\infty \xi^\kappa \frac{1}{\sigma_\mu}\,\widetilde{\mathcal{F}}_\ell\!\left(\!\dfrac{\xi}{\sigma_\mu}\right) \mathrm{d}\xi = \mathcal{F}_{\ell;\kappa} \sigma_\mu^\kappa + \mathcal{O}\!\left( \!\dfrac{1}{\sigma_\mu}\!\right)
\end{equation}
as $\sigma_\mu \rightarrow \infty$, given $\kappa =0,1,2,3,\dots$, where the leading term is independent of the value of $a$.

\begin{theorem}\label{theoremSum}
Let the marginal distribution of light particle velocities takes the form of a superposition of simpler densities~\eqref{FSum}. Then the drift and diffusion terms in~\eqref{Thm2alphabeta} reduce to
\begin{subequations}\label{Sumalphabeta}
\begin{align}
\boldsymbol{\alpha}(\mathbf{X},\mathbf{V}) & = - \sqrt{2\pi}\left( \sum_{\ell=1}^{\ell^*}\mathcal{F}_{\ell;1}\right) \gamma \left\lbrace \mathbf{V} - \mathbf{u}(\mathbf{X}) - \boldsymbol{A}[\mathbf{u}(\mathbf{y})-\mathbf{u}(\mathbf{X})]\right\rbrace + \mathcal{O}\!\left(\! \dfrac{1}{\sigma_\mu}\!\right),\\
\boldsymbol{\beta}(\mathbf{X},\mathbf{V}) & = \sqrt{\sqrt{\dfrac{\pi}{2}} \sum_{\ell=1}^{\ell^*}\mathcal{F}_{\ell;3}}\,\sqrt{2 D}\, \gamma \, \mathbf{I} + \mathcal{O}\!\left( \!\dfrac{1}{\sigma_\mu}\!\right) \quad \text{as} \quad \sigma_\mu \rightarrow \infty\,.
\end{align}
\end{subequations}
\end{theorem}
\begin{proof}
Using \eqref{FSum} we can integrate over the two inert coordinates in \eqref{mdfIntegral}, constructing the marginal distribution function according to \eqref{mdfForm}, in the form
$
\mathcal{F}(q) = \sum_{\ell=1}^{\ell^*} c_\ell \,  \widetilde{\mathcal{F}}_\ell(q).
$
From here we calculate the asymptotic scalings needed to describe the drift and diffusion terms, $\mathcal{F}_1 = \sum_{\ell=1}^{\ell^*}\mathcal{F}_{\ell;1}$ and $\mathcal{F}_3 = \sum_{\ell=1}^{\ell^*}\mathcal{F}_{\ell;3}$. The result then follows from an application of Theorem~\ref{theorem2}.
\end{proof}

\begin{example} {\rm
Consider a generalised Gaussian distribution for the light particle velocities with polynomial factors that influence the shape of the distribution (for applications of such distributions in plasma physics, see \cite{cairns1995electrostatic, izacard2016kinetic, izacard2017generalized}), say \eqref{FSum} with each $F_{\mu,\ell}(\mathbf{q})$ chosen so that the marginal distribution
$\mathcal{F}(q)$ takes the form
\begin{equation}\label{GenMBSum}
\mathcal{F}(q) = \sum_{\ell =1}^{\ell^*} c_\ell \, \dfrac{q^{2\ell}\exp(-q^2/2)}{2^{\ell + \frac{1}{2}}\,\Gamma\!\left( \ell + \frac{1}{2}\right)} 
\end{equation}
with the coefficients satisfying $\sum_{\ell=0}^{\ell^*} c_\ell = 1$. We find for \eqref{GenMBSum} that the leading-order terms in the moments \eqref{mdfxnScale} take the form 
\begin{align}
\mathcal{F}_0 = \dfrac{1}{2} \,, \quad \mathcal{F}_1 = \sum_{\ell=1}^{\ell^*}  \dfrac{c_\ell \, \Gamma(\ell +1)}{\sqrt{2} \, \Gamma\!\left( \ell + \frac{1}{2}\right)} \,, \quad \mathcal{F}_2 = \sum_{\ell=1}^{\ell^*}  \dfrac{c_\ell\,\Gamma\left( \ell + \frac{3}{2}\right)}{\Gamma\!\left( \ell + \frac{1}{2}\right)}\,, \quad \mathcal{F}_3 = \sum_{\ell=1}^{\ell^*}  \dfrac{\sqrt{2}\,c_\ell\,\Gamma\!\left( \ell + 2\right)}{\Gamma\!\left( \ell + \frac{1}{2}\right)}\,.
\end{align}
Applying Theorem \ref{theoremSum}, the drift and diffusion terms in \eqref{Thm2alphabeta} reduce to
\begin{subequations}\label{GenMBSumalphabeta}
\begin{align}
\boldsymbol{\alpha}(\mathbf{X},\mathbf{V}) & = - \sqrt{\pi}\left( \sum_{\ell=1}^{\ell^*}  \dfrac{c_\ell\,\Gamma(\ell +1)}{\Gamma\!\left( \ell + \frac{1}{2}\right)} \right) \gamma \left\lbrace \mathbf{V} - \mathbf{u}(\mathbf{X}) - \boldsymbol{A}[\mathbf{u}(\mathbf{y})-\mathbf{u}(\mathbf{X})]\right\rbrace + \mathcal{O}\!\left(\! \dfrac{1}{\sigma_\mu}\!\right)\,,\\
\boldsymbol{\beta}(\mathbf{X},\mathbf{V}) & = \sqrt{\sqrt{\pi} \sum_{\ell=1}^{\ell^*}  \dfrac{c_\ell \, \Gamma( \ell + 2)}{\Gamma\!\left( \ell + \frac{1}{2}\right)}}\,\sqrt{2 D} \,\gamma \,\mathbf{I} + \mathcal{O}\!\left(\! \dfrac{1}{\sigma_\mu}\!\right) \quad \text{as} \quad \sigma_\mu \rightarrow \infty\,.
\end{align}
\end{subequations} }
\end{example}

\subsection{Finite-size effects and the role of the flow geometry}

In this subsection we consider how the drift term~\eqref{Thm2alpha} in Theorem \ref{theorem2} behaves under stationary flows $\mathbf{u}$ of a given geometric structure. We first derive a general result for flows which are smooth enough. It is given as Theorem~\ref{theoremGeometry}, which is applied to several specific choices of flows in the corollaries. 
Theorem~\ref{theoremGeometry} can be applied to a number of flows commonly used in the fluid mechanics literature; see \cite{batchelor1967introduction} and \cite{landau2013fluid} for many examples and additional motivation behind some of the flows we study in this subsection.

\begin{theorem}\label{theoremGeometry}
Let $\mathbf{u}\in C^2(\mathbb{R}^3)$. Then, the finite-size correction \eqref{A_correction} in the drag term \eqref{Thm2alpha} reads
\begin{equation}\label{correctionFormula}
\boldsymbol{A}[\mathbf{u}(\mathbf{y})-\mathbf{u}(\mathbf{X})] = \dfrac{R^2}{10}\begin{pmatrix}
3\dfrac{\partial^2 u_1}{\partial y_1^2} + \dfrac{\partial^2 u_1}{\partial y_2^2}  + \dfrac{\partial^2 u_1}{\partial y_3^2}+ 2\dfrac{\partial^2 u_2}{\partial y_1 \partial y_2}  + 2\dfrac{\partial^2 u_3}{\partial y_1 \partial y_3} \\  
\dfrac{\partial^2 u_2}{\partial y_1^2} + 3\dfrac{\partial^2 u_2}{\partial y_2^2}  + \dfrac{\partial^2 u_2}{\partial y_3^2}+ 2\dfrac{\partial^2 u_1}{\partial y_1 \partial y_2}  + 2\dfrac{\partial^2 u_3}{\partial y_2 \partial y_3} \\  
\dfrac{\partial^2 u_3}{\partial y_1^2} + \dfrac{\partial^2 u_3}{\partial y_2^2}  + 3\dfrac{\partial^2 u_3}{\partial y_3^2}+ 2\dfrac{\partial^2 u_1}{\partial y_1 \partial y_3}  + 2\dfrac{\partial^2 u_2}{\partial y_2 \partial y_3} \\  
\end{pmatrix}_{\!\mathbf{y}=\mathbf{X}(t)}
{\hskip -8mm} + \;\, \mathcal{O}\!\left( \!R^4\!\right),
\end{equation}
as $R\rightarrow 0$.
\end{theorem}
\begin{proof}
Using the Taylor expansion near $\mathbf{y}=\mathbf{X}(t)$, we have for $\mathbf{y}\in \mathcal{S}(\mathbf{X}(t),R)$
\begin{align}
\mathbf{u}(\mathbf{y}) - \mathbf{u}(\mathbf{X}(t)) & = \sum_{\ell =1}^3 \dfrac{\partial  \mathbf{u}(\mathbf{X}(t))}{\partial y_\ell} (y_\ell - X_\ell(t)) +  \sum_{\ell =1}^3\sum_{k=1}^3 \dfrac{1}{2} \dfrac{\partial^2  \mathbf{u}(\mathbf{X}(t))}{\partial y_\ell \partial y_k} (y_\ell - X_\ell(t))(y_k - X_k(t))
\nonumber
\\
& \quad + \sum_{\ell =1}^3\sum_{k=1}^3\sum_{j=1}^3 \dfrac{1}{3} \dfrac{\partial^3  \mathbf{u}(\mathbf{X}(t))}{\partial y_\ell \partial y_k \partial y_k} (y_\ell - X_\ell(t))(y_k - X_k(t))(y_j-X_j(t)) + \mathcal{O}\!\left(\! R^4\!\right)
\nonumber
\\
& = \text{linear terms} + \text{quadratic terms} + \text{cubic terms} + \mathcal{O}\!\left(\! R^4\!\right).
\label{quadterms}
\end{align}
Since $\boldsymbol{A}_R$ is a linear operator, we have
$$
\boldsymbol{A}_R[\mathbf{u}(\mathbf{y}) - \mathbf{u}(\mathbf{X}(t))] = \boldsymbol{A}_R[\text{linear terms}] + \boldsymbol{A}_R[\text{quadratic terms}] + \boldsymbol{A}_R[\text{cubic terms}] + \mathcal{O}\! \left( \! R^4\!\right).
$$
Utilizing the symmetry of the integration domain, we have that
\begin{align*}
\int_{\mathcal{S}(\mathbf{X}(t),R)} \dfrac{\mathbf{y}-\mathbf{X}(t)}{R}\left\lbrace (y_\ell-X_\ell(t)) \mathbf{C} \cdot \dfrac{\mathbf{y}-\mathbf{X}(t)}{R}\right\rbrace \mathrm{d}A =0 
\end{align*}
for all $\ell =1,2,3$ and all constant vectors $\mathbf{C}$, from which it follows $\boldsymbol{A}_R[\text{linear terms}] = \mathbf{0}$. Similarly, 
\begin{align*}
\int_{\mathcal{S}(\mathbf{X}(t),R)} \dfrac{\mathbf{y}-\mathbf{X}(t)}{R}\left\lbrace (y_\ell-X_\ell(t))(y_j-X_j(t))(y_k-X_k(t)) \mathbf{C} \cdot \dfrac{\mathbf{y}-\mathbf{X}(t)}{R}\right\rbrace \mathrm{d}A =0 
\end{align*}
for all $\ell,j,k =1,2,3$ and all constant vectors $\mathbf{C}$, from which it follows $\boldsymbol{A}_R[\text{cubic terms}] = \mathbf{0}$. Consequently, we deduce
\begin{equation}
\boldsymbol{A}_R[\mathbf{u}(\mathbf{y}) - \mathbf{u}(\mathbf{X}(t))] = \boldsymbol{A}_R[\text{quadratic terms}] + \mathcal{O}\!\left(\! R^4\!\right),
\label{ARquad}
\end{equation}
where the quadratic terms given in~(\ref{quadterms}) will result in non-zero contributions. We consider the action of $\boldsymbol{A}_R$ on a generic vector with strictly quadratic terms, of the form 
\begin{equation}
\mathbf{B} = \sum_{\ell=1}^3\sum_{j=1}^3 \mathbf{B}_{\ell,j} (y_\ell - X_\ell(t)) (y_j - X_j(t))\,,
\label{genericform}
\end{equation}
where the $\mathbf{B}_{\ell,j}$ are constant in $\mathbf{y}$. We have that 
\begin{align}
\boldsymbol{A}_R[\mathbf{B}] & = \dfrac{3}{4\pi R^2}\int_{\mathcal{S}(\mathbf{X}(t),R)} \dfrac{\mathbf{y}-\mathbf{X}(t)}{R}\left\lbrace \left[\sum_{\ell=1}^3\sum_{j=1}^3 \mathbf{B}_{\ell,j} (y_\ell - X_\ell(t)) (y_j - X_j(t))\right] \cdot \dfrac{\mathbf{y}-\mathbf{X}(t)}{R}\right\rbrace \mathrm{d}A 
\nonumber\\
& = \dfrac{3}{4\pi R^2}\sum_{\ell=1}^3\sum_{j=1}^3\int_{\mathcal{S}(\mathbf{X}(t),R)} \dfrac{\mathbf{y}-\mathbf{X}(t)}{R} (y_\ell - X_\ell(t)) (y_j - X_j(t))\left\lbrace \mathbf{B}_{\ell,j} \cdot \dfrac{\mathbf{y}-\mathbf{X}(t)}{R}\right\rbrace \mathrm{d}A 
\nonumber 
\\
& = \dfrac{R^2}{5}\begin{pmatrix}
3B_{1,1;1} + B_{2,2;1} + B_{3,3;1} + B_{1,2;2} + B_{2,1;2} + B_{1,3;3} + B_{3,1;3} \\
B_{1,2;1} + B_{2,1;1} + B_{1,1;2} + 3B_{2,2;2} + B_{3,3;2} + B_{2,3;3} + B_{3,2;3} \\
B_{1,3;1} + B_{3,1;1} + B_{2,3;2} + B_{3,2;2} + B_{1,1;3} + B_{2,2;3} + 3B_{3,3;3}
\end{pmatrix}.
\label{ARcalc}
\end{align}
Here by $B_{\ell,j;k}$ we mean the $k$-th component ($k=1,2,3$) of $\mathbf{B}_{\ell,j}$. Comparing the general vector~$\mathbf{B}$ with the form of the quadratic terms in the expansion~(\ref{quadterms}) for $\mathbf{u}(\mathbf{y}) - \mathbf{u}(\mathbf{X}(t))$, the quadratic terms may be written in the form~(\ref{genericform}) provided
\begin{equation}
B_{\ell,j;k} = \dfrac{1}{2}\dfrac{\partial^2 u_k(\mathbf{X}(t))}{\partial y_\ell \partial y_j}\,.
\label{derterm}
\end{equation}
Since $\mathbf{u}(\mathbf{y})$ is assumed smooth, by Clairaut's Theorem we can exchange partial derivatives in~(\ref{derterm}) to get $B_{\ell,j;k} = B_{j,\ell;k}$  for all $\ell,j=1,2,3$. Substituting~(\ref{derterm})
into~(\ref{ARcalc}) and (\ref{ARquad}), we obtain~(\ref{correctionFormula}).\end{proof}

\begin{corollary}\label{CorrLinear}
If the flow of solvent particles is governed by a linear velocity field $\mathbf{u}(\mathbf{y})$ then the correction term $\boldsymbol{A}_R[\mathbf{u}(\mathbf{y})-\mathbf{u}(\mathbf{X})]$ in \eqref{Thm2alphabeta} vanishes, in which case the It\^{o} stochastic differential equation \eqref{SDEmainX}--\eqref{SDEmainY} for the motion of the heavy solute particle has drift term which takes the form 
\begin{equation}\label{CorrLinearDDterms}
\boldsymbol{\alpha}(\mathbf{X},\mathbf{V}) = - \sqrt{2\pi} \, \mathcal{F}_1 \gamma \left\lbrace \mathbf{V} - \mathbf{u}(\mathbf{X}) \right\rbrace + \mathcal{O}\!\left( \!\dfrac{1}{\sigma_\mu}\!\right) \quad \text{as} \quad \sigma_\mu \rightarrow \infty\,.
\end{equation}
\end{corollary}
\begin{proof}
For any linear flow, the formulas~(\ref{ARquad}) and~(\ref{derterm}) imply $\boldsymbol{A}_R[\mathbf{u}(\mathbf{y})-\mathbf{u}(\mathbf{X})]  = \mathbf{0}$ and the result~(\ref{CorrLinearDDterms}) then directly follows from~(\ref{Thm2alpha}).
\end{proof}

\noindent
Examples of flows for which Corollary~\ref{CorrLinear} is applicable include (i) uniform flow, $\mathbf{u}(\mathbf{x}) =\mathbf{C}$, where $\mathbf{C}$ is a constant vector; (ii) planar Couette flow, $\mathbf{u}(\mathbf{x}) = \left( 0,0,\overline{u} \, x_1\right)^\mathrm{T} $, with velocity directed along the $x_3$ axis and $\overline{u}$ being a constant; and (iii) generic stagnation flows given in the form $\mathbf{u}(\mathbf{x}) = \left( \overline{u}_1 \, x_1 , \, \overline{u}_2 \, x_2 , \, \overline{u}_3 \, x_3\right)^\mathrm{T}$, where the condition $\overline{u}_1 + \overline{u}_2 + \overline{u}_3 =0$ holds \cite{taylor1934formation, moffatt1994stretched}.
For these flows, Corollary~\ref{CorrLinear} implies that the SDEs given by drift and diffusion terms \eqref{CorrLinearDDterms} are exact, with no finite-size corrections needed. We now provide an example of a quadratic flow, the Poiseuille flow within a circular pipe \cite{sutera1993history}, to illustrate the role of the correction term.

\begin{corollary}\label{CorrPoiseuille}
If the flow of solvent particles is governed by a Poiseuille flow within a pipe of radius $h$ taking the form 
\begin{equation}\label{flow_Poiseuille}
\mathbf{u}(\mathbf{x}) = \left( 0,0,
\left( 1-\dfrac{x_1^2+x_2^2}{h^2} \right)
\overline{u}\right)^{\!\mathrm{T}},
\end{equation}
where $\overline{u}$ is a constant,
then the It\^{o} stochastic differential equation \eqref{SDEmainX}--\eqref{SDEmainY} for the motion of the heavy solute particle has drift term of the form 
\begin{equation}\label{CorrPoiseuilleDDterms}
\boldsymbol{\alpha}(\mathbf{X},\mathbf{V}) = - \sqrt{2\pi} \, \mathcal{F}_1 \, \gamma \left\lbrace \mathbf{V} - 
\left( 0,0,\left( 1- \dfrac{2R^2}{5h^2} -\dfrac{X_1^2+X_2^2}{h^2} \right) \overline{u} \right)^{\!\mathrm{T}}\right\rbrace \, + \, \mathcal{O}\!\left(\! \dfrac{1}{\sigma_\mu}\!\right)
\end{equation}
as $\sigma_\mu \rightarrow \infty$.
\end{corollary}
\begin{proof}
We have
\begin{equation*}
\dfrac{\partial^2 u_3}{\partial y_1^2} = -\dfrac{2\overline{u}}{h^2}\,, \quad \dfrac{\partial^2 u_3}{\partial y_2^2} = -\dfrac{2\overline{u}}{h^2}\,, \quad \text{and} \quad \dfrac{\partial^2 u_3}{\partial y_\ell y_j} = 0 \quad \text{otherwise}\,.
\end{equation*}
The correction term $\boldsymbol{A}[\mathbf{u}(\mathbf{y})-\mathbf{u}(\mathbf{X})]$ then reads
$$
\boldsymbol{A}[\mathbf{u}(\mathbf{y})-\mathbf{u}(\mathbf{X})] = \dfrac{R^2}{10}\left( 
0 , 
0 ,  
-\dfrac{4\overline{u}}{h^2}\right)^{\!\mathrm{T}}
=
-\dfrac{2R^2\overline{u}}{5h^2} \left( 0, 0, 1 \right)^{\mathrm{T}}\,.
$$
Using this correction term in~\eqref{Thm2alpha}, and noting that there are no higher-order corrections since the flow is exactly quadratic, we obtain~(\ref{CorrPoiseuilleDDterms}). 
\end{proof}

\noindent
As our last example in this subsection, we consider a boundary layer flow over a stretching plate, which admits an exact solution found by Crane~\cite{crane1970flow}. Assume a plate lying on the $x_3 =0$ plane is stretched along the $x_1$ axis. Then, in the region $x_1>0$, $x_2 \in \mathbb{R}$, $x_3 >0$ we have the velocity field
\begin{equation}\label{BLsoln}
\mathbf{u}(\mathbf{x}) = \left( \overline{u} \, \dfrac{x_1}{h_1} \exp\!\left( \!- \dfrac{x_3}{h_3}\!\right), \, 0 , \, - \frac{\overline{u}\, h_3}{h_1}\left[ 1- \exp\!\left( \!- \dfrac{x_3}{h_3}\!\right)\right]\right)^{\!\mathrm{T}},
\end{equation}
where $\overline{u},$ $h_1$ and $h_3$ are positive constants.

\begin{corollary}\label{CorrBoundaryLayer}
If the flow of solvent particles is governed by a boundary layer flow described by the solution of Crane~\eqref{BLsoln}, then the It\^{o} stochastic differential equation \eqref{SDEmainX}--\eqref{SDEmainY} for the motion of the heavy solute particle has drift term of the form 
\begin{equation}
\label{CorrBLDDterms}
\boldsymbol{\alpha}(\mathbf{X},\mathbf{V}) = - \sqrt{2\pi}\,\mathcal{F}_1 \, \gamma \left\lbrace \mathbf{V} - \overline{u}\begin{pmatrix}
\left[1+\dfrac{R^2}{10h_3^2}\right]\dfrac{X_1(t)}{h_1} \, \mathrm{e}^{-X_3(t)/h_3}\\ 0\\ \dfrac{h_3}{h_1}\left[1+\dfrac{R^2}{10h_3^2}\right] \mathrm{e}^{-X_3(t)/h_3}- \dfrac{h_3}{h_1} \end{pmatrix} \right\rbrace \, + \, \mathcal{O}\!\left( \!\dfrac{1}{\sigma_\mu}\!\right) 
\, + \, \mathcal{O}\!\left(\!R^4\!\right)
\end{equation}
as $\sigma_\mu \rightarrow \infty$ and $R\rightarrow 0$.
\end{corollary}
\begin{proof}
Differentiating the components of~(\ref{BLsoln}), we obtain
\begin{equation*}
\dfrac{\partial^2 u_1}{\partial y_1 \partial y_3} 
= 
- \dfrac{\overline{u}}{h_1 h_3}\mathrm{e}^{-y_3/h_3}, \qquad \dfrac{\partial^2 u_1}{\partial y_3^2} 
= 
\dfrac{\overline{u}}{h_1 h_3^2} \, y_1 \, \mathrm{e}^{-y_3/h_3}, 
\qquad \dfrac{\partial^2 u_3}{\partial y_3^2} = \dfrac{\overline{u}}{h_1 h_3} \, \mathrm{e}^{-y_3/h_3}, 
\end{equation*}
while all other second-order partial derivatives are zero. Using Theorem~\ref{theoremGeometry}, we have
\begin{align*}
\boldsymbol{A}_R[\mathbf{u}(\mathbf{y})-\mathbf{u}(\mathbf{X})]
+
\mathcal{O}\!\left( \!R^4\!\right) 
= \dfrac{R^2}{10}\begin{pmatrix}
 \dfrac{\partial^2 u_1}{\partial y_3^2}
 \\  
0 \\  
3\dfrac{\partial^2 u_3}{\partial y_3^2}+ 2\dfrac{\partial^2 u_1}{\partial y_1 \partial y_3} 
\\  
\end{pmatrix}_{\! \mathbf{y}=\mathbf{X}(t)}
{\hskip -1cm} = 
\dfrac{R^2 \, \overline{u} \, \mathrm{e}^{-X_3(t)/h_3}}{10 h_1 h_3}
\begin{pmatrix}
 X_1(t)/h_3\\  
0 \\  
1 \\  
\end{pmatrix}.
\end{align*}
Substituting into~\eqref{Thm2alpha},
we obtain~(\ref{CorrBLDDterms}).
\end{proof}

\section{Simulations in a co-moving frame}
\label{sec4}

In this section, we validate the results of Section~\ref{sec3} by performing illustrative computer simulations. One way to compare the theory and simulation is to numerically approximate the scale factors for the drift ($\chi_1(\theta)$) and diffusion ($\chi_2(\theta)$) terms in the motion of a heavy particle predicted by the theory, assuming the velocities of solvent particles are sampled according to the generalized Gaussian distribution~(\ref{gengauss}). Since we want to estimate the scale factors $\chi_1(\theta)$ and $\chi_2(\theta)$, we need to perform simulations over sufficiently long time, averaging over many collisions with solvent particles. Our simulations will consider more than $10^8$ solvent particles and the system will evolve for more than $10^8$ time steps of length $\Delta t$. In particular, a direct simulation of such a large system would be computationally intensive. 

To design an efficient computational scheme, we will only explicitly simulate the behaviour of about $900$ particles at any one time, because our computer simulations will be performed in a co-moving cubic frame of length $L$ that is centered at $\mathbf{X}(t)$, i.e. solvent particles are explicitly simulated in the cubic box
\begin{equation}
\mathbf{X}(t)
+
\left[
-
\frac{L}{2}, 
\frac{L}{2}
\right]
\times
\left[
-
\frac{L}{2}, 
\frac{L}{2}
\right]
\times
\left[
-
\frac{L}{2}, 
\frac{L}{2}
\right],
\label{coframe}
\end{equation} 
while we assume that the solvent particles are distributed according to the spatial Poisson process with density $\lambda_{\mu}$ given by~(\ref{lambda3Dexp}) and their velocities are distributed according to $f_\mu({\mathbf v},{\mathbf u})$ given by~(\ref{Fgeneral}). Simulation in a co-moving frame was introduced in~\cite{gunaratne2019multi} for the case without flow, i.e. for  ${\mathbf u} \equiv {\mathbf 0}$ and for the Gaussian distribution~(\ref{fvel3Dexp}). Here, we extend this approach to the case of general flow, i.e. for ${\mathbf u} \ne {\mathbf 0}$ and for general velocity distribution (\ref{Fgeneral}). One iteration (that is an update of the state of the system from time $t$ to time $t+\Delta t$) for the simulation in co-moving frame~(\ref{coframe}) is given as Algorithm [A1]--[A6] in Table~\ref{table2}. It evolves
the position and velocity of the solute particle together with the positions
and velocities of $N(t)$ solvent particles, where $N(t)$ depends on time $t$, and $N(t)$ will fluctuate around the value $900$ in our presented illustrative simulations. To formulate Algorithm [A1]--[A6], we assume that the timestep $\Delta t$ is chosen small enough so that at most one collision happens per iteration. 

\begin{table}
\caption{\label{table2}
{\it One iteration of the simulation algorithm in a co-moving frame.}}
\framebox{%
\hsize=0.961\hsize
\vbox{
\leftskip 9.8mm
\parindent -9.8mm

[A1] \hskip 1.2mm
Update the positions of the solute and solvent particles by calculating their free-flight positions~(\ref{frfl1})--(\ref{frfl2}).

\smallskip

[A2] \hskip 1.2 mm Update the velocities of all solvent particles for which the free-flight position~(\ref{frfl2}) lies outside the radius 
of the solute particle by equation~(\ref{discvel}).

\smallskip

[A3] \hskip 1.2mm
If the free-flight position~(\ref{frfl2}) of a solvent particle lies within the radius 
of the solute particle, reverse the trajectories of the solvent and the solute particles by time 
$\tau < \Delta t$ such that they are just touching. Calculate post-collision velocities 
by equations~(\ref{elcol1})--(\ref{elcol2}) and further update their new positions by moving forward by time $\tau$. Otherwise, each free-flight position~(\ref{frfl2})
is accepted as the particle's position at time $t+\Delta t$.
 
\smallskip

[A4] \hskip 1mm
Update $N(t)$ by removing those solvent particles which lie outside of the co-moving frame~(\ref{coframe}) centered at new position of the solute particle, $\mathbf{X}(t+\Delta t)$, from the simulation. Calculate the velocity of the co-moving frame, $\mathbf{V}_{\mbox{\scriptsize f}}$,  
over the time interval $[t,t+\Delta t]$ by equation~(\ref{framevel}). 

\smallskip

[A5] \hskip 1mm
Generate six Poisson distributed random numbers $N_i^-$ and $N_i^+$, for $i=1,2,3$, with means $p_i^{-}({\mathbf V}_{\!\mbox{\scriptsize f}})$ and $p_i^{+}({\mathbf V}_{\!\mbox{\scriptsize f}})$,
respectively. For each side $i$, generate $N_i^-$ and $N_i^+$ proposed positions and velocities of new solvent particles. For each proposed new particle position ${\mathbf x}_{\mbox{\scriptsize new}}$ 
with velocity ${\mathbf v}_{\mbox{\scriptsize new}}$, generate a random number $r$ uniformly distributed in interval $(0,1).$ \hfill\break
If $r < h_{\mbox{\scriptsize acc}}({\mathbf x}_{\mbox{\scriptsize new}},{\mathbf v}_{\mbox{\scriptsize new}})$, where $h_{\mbox{\scriptsize acc}}({\mathbf x}_{\mbox{\scriptsize new}},{\mathbf v}_{\mbox{\scriptsize new}})$ is given by~(\ref{acceptprob}), then increase $N(t)$ by 1 and initialize the new solvent particle at the proposed position ${\mathbf x}_{\mbox{\scriptsize new}}$ 
with velocity ${\mathbf v}_{\mbox{\scriptsize new}}$.

\smallskip

[A6] \hskip 1.2mm
Continue with step [A1] using time $t=t+\Delta t$.

\par \vskip 0.8mm}
} 
\end{table}

We initialize the position of the solute particle as $X(0)=0$ and generate a Poisson number (with mean $\lambda_{\mu} \, L^3$) of solvent particles in our simulation domain~(\ref{coframe}). The solvent particles are initially placed uniformly in the frame (\ref{coframe}), where we remove particles overlapping with the solute particle (before we begin our simulation) to get the initial number, $N(0)$, of simulated solvent particles. Their initial velocities are drawn from the distribution $f_\mu({\mathbf v},{\mathbf u})$ given by~(\ref{Fgeneral}).

We first update the positions in Step [A1] over the time interval $(t,t+\Delta t]$ using the free-flight equations
\begin{eqnarray}
{\widehat{{\mathbf X}}}(t+\Delta t) &=& {\mathbf X}(t) + {\mathbf V}(t) \, \Delta t, 
\label{frfl1} \\
{\widehat{{\mathbf x}}_j}(t+\Delta t) &=& {\mathbf x}_j(t) + {\mathbf v}_j(t) \, \Delta t,
\label{frfl2}
\end{eqnarray}
where $j=1,2,\dots,N(t)$. Since $\Delta t$ is chosen small enough so that only one collision happens during the time interval
$[t,t+\Delta t]$, most of the free-flight positions~(\ref{frfl2}) of solvent particles are accepted 
in Step~[A3] as their updated positions ${\mathbf x}_j(t+\Delta t)$ and only the position of the solvent particle colliding with the solute particle is further updated. In Step~[A2], we update the velocities of the solvent particles which did not collide with the solute particle by a discretized version of equation~(\ref{xveqgen}), namely by
\begin{equation}
{\mathbf v}_j(t+\Delta t) 
=
{\mathbf v}_j(t) + \left(
\nabla {\mathbf u} (\mathbf{x}_j(t)) 
\right)
{\mathbf v}_j(t)
\, \Delta t,
\label{discvel}    
\end{equation}
where $\nabla {\mathbf u} (\mathbf{x}_j(t))$ means that we evaluate $\nabla {\mathbf u}$ at the position $\mathbf{x}_j(t)$ of the $j$-th particle at time $t$. This means that equations~(\ref{frfl2}) and (\ref{discvel}) form the forward Euler discretization of ODEs~(\ref{xveqgen}). To determine the colliding particle, we evaluate the condition
$$
\left\|
{\widehat{{\mathbf x}}_j}(t+\Delta t)
-
{\widehat{{\mathbf X}}}(t+\Delta t)
\right\|_2
< 
R
$$
for all solvent particles in steps [A2] and [A3]. If the particle did collide, then we do not use equation~(\ref{discvel}) during the collision time step. In Step~[A4], we update the position and velocity of the co-moving frame to ${\mathbf X}(t+\Delta t)$ and 
\begin{equation}
{\mathbf V}_{\mbox{\scriptsize f}} 
= 
\frac{{\mathbf X}(t+\Delta t) - {\mathbf X}(t)}{\Delta t}.
\label{framevel}
\end{equation}
We remove solvent particles which are outside of the simulation domain and update $N(t)$ accordingly. In Step~[A5], we calculate the number of particles which entered the simulation domain during the time interval $(t,t+\Delta t]$ through each corresponding side of the co-moving frame. The mean number of particles entering through the left ($-$) and right ($+$) of the $i$-th side is denoted as 
\begin{equation}
p_i^{\pm}({\mathbf V}_{\!\mbox{\scriptsize f}}), 
\qquad
\mbox{for} \quad i=1,2,3,
\label{incomeprob}
\end{equation}
and it is calculated for each considered distribution in Section~\ref{secdistincome}, where we also specify the distributions of the initial positions ${\mathbf x}_{\mbox{\scriptsize new}}$ and velocities ${\mathbf v}_{\mbox{\scriptsize new}}$ of the introduced solvent particles. To derive the corresponding equations, we integrate over the half-space $(-\infty,0) \times \mathbb{R}^2$, meaning that once we consider all six faces of the cubic frame~(\ref{coframe}) we have over-counted twice at the edges and three times at the corners.
To compensate for this effect, we consider the sampled position,
${\mathbf x}_{\mbox{\scriptsize new}}$ and velocity ${\mathbf v}_{\mbox{\scriptsize new}}$ of the new incoming
particle at time $t+\Delta t$ and calculate its previous position at time $t$ by
$$
{\mathbf y}
=
{\mathbf x}_{\mbox{\scriptsize new}} - {\mathbf v}_{\mbox{\scriptsize new}} \, \Delta t.
$$
If ${\mathbf y}$ is in the regions which were counted twice or three times in
our derivation, we reject the proposed introduction of the 
new solvent particle with the corresponding probability.
Namely, we use the acceptance probability in Step~[A5] given by 
\begin{equation}
h_{\mbox{\scriptsize acc}}({\mathbf x}_{\mbox{\scriptsize new}},{\mathbf v}_{\mbox{\scriptsize new}})
=
\left\{
\begin{array}{ll}
1,
&
\mbox{ for}\;\; \; {\mathbf y}-\mathbf{X} (t) \in {\mathcal Y}_1;
\\
1/2,
&
\mbox{ for}\;\; \; {\mathbf y}-\mathbf{X} (t) \in {\mathcal Y}_2; 
\\
1/3,
&
\mbox{ for}\;\; \; {\mathbf y}-\mathbf{X} (t)\in {\mathcal Y}_3,
\\
\end{array}
\right.
\label{acceptprob}
\end{equation}
where ${\mathcal Y}_j \subset {\mathbb R}^3$ is the region of the space
which consists of points which have exactly $j$ of their coordinates outside of the interval $[-L/2,L/2]$. Regions ${\mathcal Y}_1$ and ${\mathcal Y}_2$ are illustrated in the schematic of our co-moving frame in Figure~\ref{fig1}(a) with some of their boundaries highlighted by the magenta and green shading.
Region ${\mathcal Y}_3$ is next to the corners of the cube consisting of points, which have all coordinates outside of the interval $[-L/2,L/2]$.

\begin{figure}
(a) \hskip 8cm (b) \hfill\break
\rule{0pt}{1pt}
\hskip 4mm \epsfig{file=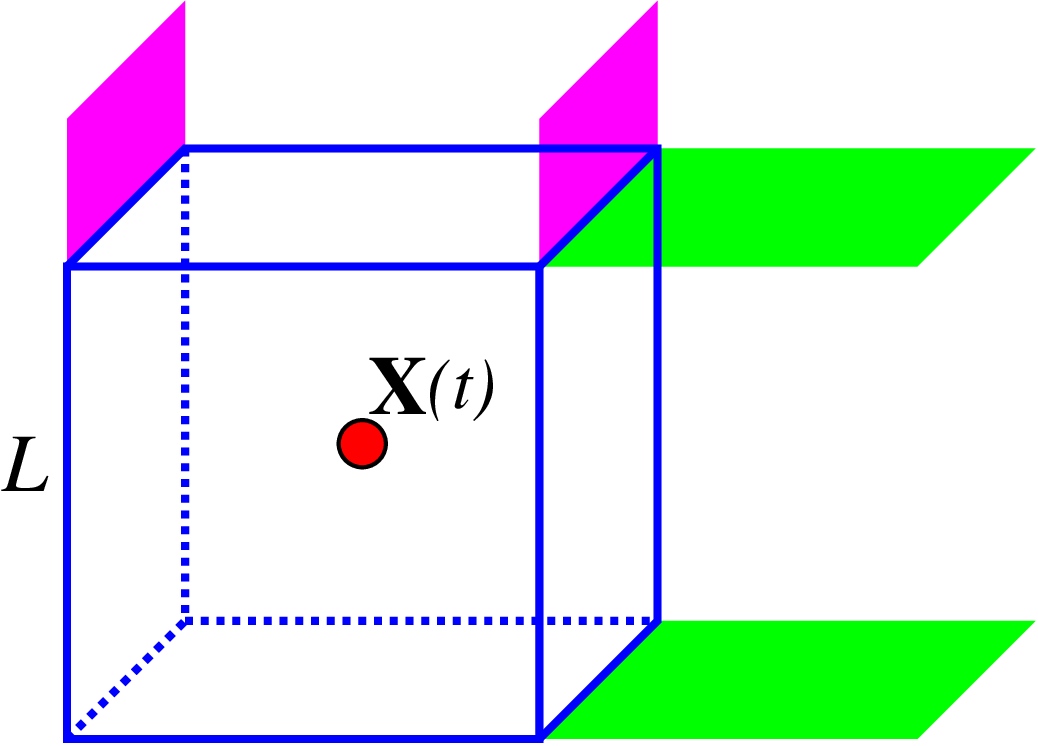,height=5cm} \hskip 8mm \raise -4mm \hbox{\epsfig{file=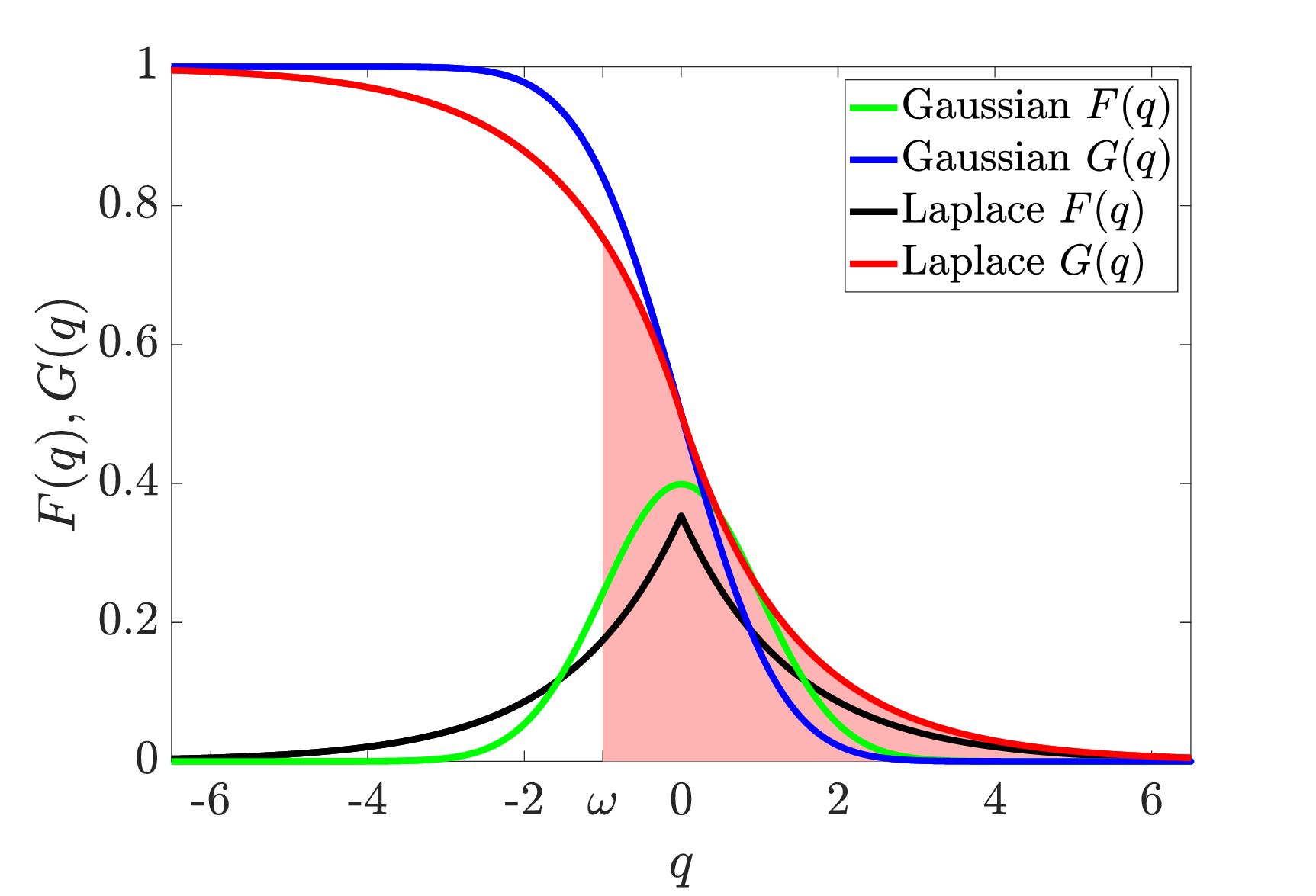,height=5.9cm}}
\caption{\label{fig1} (a) 
{\it Schematic of the co-moving frame~$(\ref{coframe})$ is shown as the blue cube with centre ${\mathbf X}(t)$ and size $L$. Some parts of region ${\mathcal Y}_1$ $($used in equation~$(\ref{acceptprob}))$ are between the magenta shaded planes and between the green shaded planes, while the region $($next to the edge$)$ bounded by the green and magenta planes is a part of ${\mathcal Y}_2$.}
(b) {\it Functions ${\mathcal F}(q)$ and ${\mathcal G}(q)$ given by equation~$(\ref{functioncalgh})$ and Table~$\ref{table3}$. The shaded part illustrates the probability distribution~$(\ref{shiftcalG})$ for $\omega=-1$.}}
\end{figure}

\subsection{Distributions of particles entering the co-moving frame in Step~[A5]}
\label{secdistincome}

To apply Algorithm [A1]-[A6], we need to calculate the number, positions and velocities of incoming solvent particles in Step~[A5]. We will study the behaviour at the boundary side 
\begin{equation}
\left\{ X_1(t) - \frac{L}{2}\right\}
\times
\left[
X_2(t)
-
\frac{L}{2}, 
X_2(t)
+
\frac{L}{2}
\right]
\times
\left[
X_3(t) -
\frac{L}{2},
X_3(t)
+
\frac{L}{2}
\right]
\label{leftfirst}
\end{equation}
in Theorem~\ref{thmbound}, where this boundary is approximated as the plane
$\{ X_1(t) - L/2\} \times {\mathbb R}^2$ dividing the space ${\mathbb R}^3$ into two half-spaces approximating the exterior and the interior of the co-moving frame~(\ref{coframe}). We consider those particles which are not explicitly simulated at time $t$ (because they are outside of the co-moving frame at time $t$), but they have to be included in the simulation at time $t+\Delta t$, because they entered the co-moving frame at time $t+\Delta t$. The results for the five other sides of the co-moving frame~(\ref{coframe}) can be obtained in the same way. 

\begin{theorem} \label{thmbound}
Let $\Delta t > 0$. 
Let us assume that solvent particles are distributed 
according to the Poisson distribution with density 
$\lambda_\mu$ in the half space $(-\infty,X_1(t)-L/2) \times {\mathbb R}^2$; 
their initial velocities are distributed according to 
$f_\mu({\mathbf v},{\mathbf u})$ and there are no particles in 
the half space $(X_1(t)-L/2,\infty) \times {\mathbb R}^2$ 
at time $t$. Then the positions ${\mathbf x}$ and velocities 
${\mathbf v}$ of solvent particles  are distributed at time 
$t + \Delta t$ according to 
\begin{align}
g({\mathbf x},\mathbf{v}; \Delta t)
& =
H \big(X_1(t) - L/2 - x_1 + v_1 \Delta t
\big) \, \lambda_\mu \, f_\mu({\mathbf v},{\mathbf u})
\nonumber
\\
& \quad -
H \big(X_1(t) - L/2 - x_1 \big) \, \lambda_\mu \, \nabla_{\mathbf u}
f_\mu({\mathbf v},{\mathbf u})
\left(\nabla {\mathbf u} \right){\mathbf v}
\, \Delta t
+
O\left( (\Delta t)^2\right),
\label{distveldeltatbound}
\end{align}
where $H(\cdot)$ is the Heaviside step function and
$$
\nabla_{\mathbf u}
f_\mu({\mathbf v},{\mathbf u})
\left(\nabla {\mathbf u} \right){\mathbf v}
\,
\equiv
\,
\sum_{j,\ell=1}^3 
\frac{\partial f_\mu}{\partial u_j}({\mathbf v},{\mathbf u}({\mathbf x})) \,
\frac{\partial u_j}{\partial x_\ell} ({\mathbf x})\, v_{\ell} \, .
$$
The marginal distribution of the positions ${\mathbf x}$ of solvent particles at time $t = \Delta t$ is then
\begin{align}
\varrho({\mathbf x}; \Delta t)
& =
\lambda_\mu \, 
\int_{(L/2 + x_1-X_1(t))/\Delta t}^\infty
\int_{-\infty}^{\infty}\int_{-\infty}^{\infty}
f_\mu({\mathbf v},{\mathbf u})
\, \mbox{{\rm d}} v_3
\, \mbox{{\rm d}} v_2
\, \mbox{{\rm d}} v_1 
\nonumber
\\
& \quad - 
H \big(X_1(t) - L/2 - x_1 \big) \, \lambda_\mu \, \Delta t \,
\int_{{\mathbb R}^3}
\nabla_{\mathbf u}
f_\mu({\mathbf v},{\mathbf u})
\left(\nabla {\mathbf u} \right){\mathbf v}
\, \mbox{{\rm d}} {\mathbf{v}} 
+ O\left( (\Delta t)^2\right)\,.
\qquad
\label{distdeltatbound}
\end{align} 
\end{theorem}

\begin{proof}
Solvent particles which have their position ${\mathbf x} = (x_1,x_2,x_3)^{\mathrm{T}}$ and velocity ${\mathbf v}=(v_1,v_2,v_3)^{\mathrm{T}}$ at time $t + \Delta t$ were previously (at time $t$) with the position and velocity, which we denote by ${\mathbf x}_{\mbox{\scriptsize old}}=(x_{1,\mbox{\scriptsize old}},x_{2,\mbox{\scriptsize old}},x_{3,\mbox{\scriptsize old}})^{\mathrm{T}}$ and ${\mathbf v}_{\mbox{\scriptsize{old}}}
=(v_{1,\mbox{\scriptsize{old}}},v_{2,\mbox{\scriptsize{old}}},v_{3,\mbox{\scriptsize{old}}})^{\mathrm{T}}$. There will be nonzero solvent particles with velocity ${\mathbf v}$ at point 
${\mathbf x}$ at time $t + \Delta t$ provided that 
$x_{1,\mbox{\scriptsize{old}}} < X_1(t)-L/2$ which implies that 
\begin{equation}
g({\mathbf x},\mathbf{v}; \Delta t)
=
H(X_1(t) - L/2 - x_{1,\mbox{\scriptsize{old}}}
) \, \lambda_\mu \, f_\mu({\mathbf v},{\mathbf u}(\mathbf{x}_{\mbox{\scriptsize{old}}}) ).
\label{exactboundist}
\end{equation}
Using equations~(\ref{frfl2}) and (\ref{discvel}), we have
\begin{eqnarray*}
x_{1,\mbox{\scriptsize old}}
&=&
x_1 -
v_{1,\mbox{\scriptsize{old}}} \, \Delta t,
\\
v_{1,\mbox{\scriptsize{old}}} 
&=&
v_1 -
\left(
\nabla {\mathbf u} (\mathbf{x}_{\mbox{\scriptsize{old}}}) 
\right)
{\mathbf v}_{\mbox{\scriptsize{old}}}
\, \Delta t,
\end{eqnarray*}
Substituting into equation~(\ref{exactboundist}), we obtain 
$$
g({\mathbf x},\mathbf{v}; \Delta t)
=
H \big(X_1(t) - L/2 - x_1 + v_1 \Delta t
\big) \, \lambda_\mu \, f_\mu({\mathbf v},{\mathbf u}(\mathbf{x} - {\mathbf v} \Delta t))
+
O\big((\Delta t)^2\big).
$$
Using the Taylor expansion, we obtain
equation (\ref{distveldeltatbound}) to order $O((\Delta t)^2)$. The marginal distribution of the positions ${\mathbf x}$ of solvent particles at time $t + \Delta t$ can then be obtained by integrating  $g({\mathbf x},\mathbf{v}; \Delta t)$ over all possible velocities, giving
$$
\varrho({\mathbf x}; \Delta t)
=
\int_{{\mathbb R}^3}
g({\mathbf x},\mathbf{v}; \Delta t)
\, \mbox{{\rm d}} {\mathbf v}.
$$
Substituting~(\ref{distveldeltatbound}), we obtain
(\ref{distdeltatbound}).
\end{proof}

To apply Theorem~\ref{thmbound} to the left side in the first direction~(\ref{leftfirst}) of the co-moving frame~(\ref{coframe}), we identify the co-moving frame~(\ref{coframe}) at time $t+\Delta t$ with the semi-infinite cuboid 
\begin{equation}
{\mathbf X}(t+\Delta t)
+
\left[
- \frac{L}{2},\infty
\right) 
\times
\left[
-
\frac{L}{2}, 
\frac{L}{2}
\right]
\times
\left[
-
\frac{L}{2},
\frac{L}{2}
\right].
\label{semicuboid}
\end{equation}
Then the mean number of particles entering through the left side in the first direction~(\ref{leftfirst}) of the co-moving frame can be calculated as the average number of particles in the semi-infinite cuboid~(\ref{semicuboid}) 
by integrating
\begin{equation}
\label{p1min}
p_1^{-}({\mathbf V}_{\!\mbox{\scriptsize f}})
=
\int_{[X_1(t+\Delta t)-L/2,\infty)  \times [0,L]^2}
\varrho({\mathbf x}; \Delta t)
\,\mbox{{\rm d}} {\mathbf x},
\end{equation}
where $\varrho({\mathbf x}; \Delta t)$ is given by~(\ref{distdeltatbound}) and we have used that the integration domain can be translated in the second and third directions without changing the integral. In order to evaluate integrals in equations~(\ref{distdeltatbound}) and (\ref{p1min}), 
 we will assume that the velocity distribution $f_\mu({\mathbf v},{\mathbf u})
= 
F_{\mu}(\mathbf{v}-\mathbf{u})$ in equation~(\ref{Fgeneral}) is given in the product form~(\ref{productform}) and ${\mathbf u}$ is a constant vector (uniform flow). We define
\begin{equation}
{\mathcal G}(z)
=
\int_z^\infty
{\mathcal F}(q)
\, \mbox{d} q,
\qquad
\mbox{and}
\qquad
{\mathcal K}(z)
=
\int_z^\infty
{\mathcal G}(q)
\, \mbox{d} q.
\label{functioncalgh}
\end{equation}
Since $\nabla {\mathbf u}$ is a zero matrix for uniform flow ${\mathbf u}$, equation~(\ref{distdeltatbound}) simplifies as follows
\begin{equation}
\varrho({\mathbf x}; \Delta t)
\equiv
\varrho(x_1; \Delta t)
=
\lambda_\mu \, 
{\mathcal G}\!\left(\frac{
L + 2 \, x_1 - 2 \, X_1(t)
- 2 \, u_1 \, \Delta t}{2 \, \sigma_\mu \, \Delta t} \right).
\label{distdeltatboundproduct}
\end{equation}
Substituting into~(\ref{p1min}), we get
$$
p_1^{-}({\mathbf V}_{\!\mbox{\scriptsize f}})=
\lambda_\mu \, L^2 \, \sigma_\mu \, \Delta t \,\,
{\mathcal K}\!\left(\frac{V_{\mbox{\scriptsize f},1}
- u_1}{\sigma_\mu} \right)
$$
where
$
{\mathbf V}_{\!\mbox{\scriptsize f}}
=
(
V_{\mbox{\scriptsize f},1},
V_{\mbox{\scriptsize f},2},
V_{\mbox{\scriptsize f},3}
)^{\mathrm{T}}
$
is the frame velocity defined by~(\ref{framevel}) and $\mathbf{u}$ is the underlying flow vector field.
Using symmetry, we can then express the numbers of particles entering the co-moving frame during one time step through each side by
\begin{equation}
p_i^{\pm}(\mathbf{V}_f)=
\lambda_\mu \, L^2 \, \sigma_\mu \, \Delta t \,\,
{\mathcal K}\!\left(
\pm \frac{u_i - V_{\mbox{\scriptsize f},i}
}{\sigma_\mu} \right), 
\qquad
\mbox{for} \quad i=1,2,3.
\label{pivf}
\end{equation}
These numbers are then used in Step~[A5] of Algorithm~[A1]-[A6] in Table~\ref{table2}. The distribution of the initial positions of the incoming particle is proportional to~(\ref{distdeltatboundproduct}) which is restricted to the semi-infinite cuboid~(\ref{semicuboid}). Since~(\ref{distdeltatboundproduct}) only depends on $x_1,$ the distribution of the second and third coordinates is uniform in 
$[X_2(t+\Delta t)-L/2,X_2(t+\Delta t)+L/2]$ and
$[X_3(t+\Delta t)-L/2,X_3(t+\Delta t)+L/2]$, respectively, while the first coordinate can be sampled as
\begin{equation}
x_1 = X_1(t+\Delta t) - L/2 + \zeta \, \sigma_\mu \, \Delta t \,,
\label{zetadef}
\end{equation}
where $\zeta$ is the dimensionless distance (expressed in units of $\sigma_\mu \, \Delta t$) from the boundary~(\ref{leftfirst}) at time $t+\Delta t.$ Using~(\ref{distdeltatboundproduct}), we can sample $\zeta$ according to a distribution proportional to
$$
{\mathcal G}\!\left(
\zeta
+
\frac{V_{\mbox{\scriptsize f},1}
- u_1}{\sigma_\mu} \right),
\qquad
\mbox{for} \quad \zeta > 0.
$$
Using symmetry, the distances of particles entering the co-moving frame through each of its sides are sampled according to distributions proportional to
\begin{equation}
{\mathcal G}\!\left(
\zeta
\pm
\frac{u_i - V_{\mbox{\scriptsize f},i}}{\sigma_\mu} \right),
\qquad
\mbox{for} \quad \zeta > 0.
\label{distdeltatboundproductzeta}
\end{equation}
To sample random numbers from distributions~(\ref{distdeltatboundproductzeta}), we will use the acceptance-rejection algorithm [S1]-[S4] presented in Table~\ref{tableaccrej}.

\begin{table}
\caption{\label{tableaccrej} {\it Acceptance-rejection algorithm for sampling random numbers according to the probability distribution $p(\zeta;\omega)$ given by~$(\ref{shiftcalG})$.}}
\framebox{%
\hsize=0.961\hsize
\vbox{
\leftskip 8.8mm
\parindent -8.8mm

[S1] \hskip 1.2mm
Calculate $a_1(\omega)$ and $a_2(\omega)$ that satisfy the inequality~(\ref{accrejcond}) for all $\zeta > 0$.

\smallskip

[S2] \hskip 1.2 mm 
Generate two random numbers $\eta_1$ and $\eta_2$ uniformly distributed in the interval (0,1).

\smallskip

[S3] \hskip 1.2 mm
Compute an exponentially distributed random number $\eta_3$ by
$\eta_3 = - a_1(\omega) \, \log(\eta_1).$

\smallskip

[S4] \hskip 1.2mm
If $\eta_1 \, \eta_2 < a_2(\omega) \, {\mathcal G} (\eta_3 + \omega)\,$, then choose $\eta_3$ as a sample from the probability \hfill\break distribution
(\ref{shiftcalG}). Otherwise, go to step [S2] of the algorithm.
\par \vskip 0.8mm}
}
\end{table}

This is a generalization of the acceptance-rejection algorithms that were previously used for simulations with Maxwell-Boltzmann statistics~\cite{Erban:2014:MDB,gunaratne2019multi}. In the case of the distributions~(\ref{distdeltatboundproductzeta}), we need to sample random numbers according to the probability distribution
\begin{equation} 
p(\zeta;\omega) = \frac{{\mathcal G}(\zeta+\omega)}{{\mathcal K}(\omega)} \, ,
\qquad \mbox{for} \quad \zeta > 0,
\label{shiftcalG}
\end{equation}
where $\omega \in {\mathbb R}$ is a parameter. This distribution is illustrated in Figure~\ref{fig1}(b) using the pink shading for $\omega=-1$ and ${\mathcal F}$ being the Laplace distribution.

The algorithm in Table~\ref{tableaccrej} samples random numbers according to the distribution~(\ref{shiftcalG}) by generating an exponentially distributed 
random number $\eta_3$ with mean $a_1(\omega)$, which is a parameter of the method satisfying
\begin{equation}
p(\zeta;\omega)
\,\le\,
\frac{1}{a_2(\omega) \, {\mathcal K}(\omega)} 
\, \exp \! 
\left[
-
\frac{\zeta}{a_1(\omega)}
\right] 
\qquad \mbox{for all} \quad \zeta > 0,
\label{auxaccrej}
\end{equation}
where $a_2(\omega)$ is the second parameter of the method. Substituting~(\ref{shiftcalG}) into~(\ref{auxaccrej}), we get
\begin{equation}
a_2(\omega)
\, {\mathcal G}(\zeta+\omega)
\,
\exp \! 
\left[
\frac{\zeta}{a_1(\omega)}
\right] 
\,\le\,
1
\qquad \mbox{for all} \quad \zeta > 0.
\label{accrejcond}
\end{equation}
The inequality~(\ref{accrejcond}) can always be satisfied for some choices of parameters $a_1(\omega)$ and $a_2(\omega)$, because ${\mathcal G}$ is exponentially decreasing for large values of $\zeta$ as can be seen in Figure~\ref{fig1}(b). In practice, we have to choose $a_1(\omega)$ and $a_2(\omega)$ to have a relatively high acceptance probability which is the number on the left-hand side of condition~(\ref{accrejcond}). Since our exponentially distributed random number $\eta_3$ is obtained in Step [S3] as $\eta_3 = - a_1(\omega) \, \log(\eta_1),$ we can substitute this into the left-hand-side of condition~(\ref{accrejcond}) for $\zeta$ to get the acceptance probability in Step~[S4] as
\begin{equation}
\frac{a_2(\omega)
\, {\mathcal G}(\eta_3+\omega)}{\eta_1} \, .
\label{accprobS4}
\end{equation}
A relatively high acceptance probability gives an efficient algorithm, because it decreases the number of repeats of Steps~[S2]-[S4] in Table~\ref{tableaccrej}. Since $\eta_1$ is uniformly distributed in Step~[S2] and $\eta_3 = - a_1(\omega) \, \log(\eta_1),$ we can substitute into~(\ref{accprobS4}) to get the probability that the algorithm [S2]-[S4] finishes in one iteration as
\begin{equation}
\int_0^1
\frac{a_2(\omega)
\, {\mathcal G}(- a_1(\omega) \, \log(s)
+
\omega)}{s}
\, \mbox{d} s
=
\int_{\omega}^\infty
\frac{a_2(\omega)}{a_1(\omega)}
\, {\mathcal G}(\zeta)
\, \mbox{d} \zeta
=
\frac{a_2(\omega)\, {\mathcal K}(\omega)}{a_1(\omega)}.
\label{accprobS4t}
\end{equation}
An appropriate choice of the parameters $a_1(\omega)$ and $a_2(\omega)$ will depend on our choice of ${\mathcal F}$. In our illustrative simulations, we will use two of the functions ${\mathcal F}: {\mathbb R} \to [0,\infty)$ which have been given in Table~\ref{tab:mdfxnscalinglist} in Section~\ref{sec3}, namely, the Gaussian (Maxwell-Boltzmann statistics) and Laplace distributions. For these distributions, the integrals ${\mathcal G}$ and ${\mathcal K}$ defined by~(\ref{functioncalgh}) can be evaluated and are given in Table~\ref{table3} together with the choices of $a_1(\omega)$ and $a_2(\omega)$, which we use in our simulations. However, any choice of $a_1(\omega)$ and $a_2(\omega)$ will lead to correct sampling of random numbers by the algorithm [S1]-[S4] provided that they satisfy the inequality~(\ref{accrejcond}) for all $\zeta>0.$ Finally, we need a method for sampling velocities of particles introduced into the simulation in Step~[A5] according to distribution~(\ref{distveldeltatbound}). In the case of a velocity distribution $f_\mu({\mathbf v},{\mathbf u})
= 
F_{\mu}(\mathbf{v}-\mathbf{u})$ taking the form of a product as in~(\ref{productform}) and ${\mathbf u}$ being a constant vector (uniform flow), this simplifies to sampling the first coordinate of the incoming velocity of the particle entering through the first side according to the truncated distribution proportional to
\begin{equation}
 \, {\mathcal F}\!\left(\frac{v_1-u_1}{\sigma_\mu} \right)
\qquad
\mbox{restricted to the subdomain}
\quad
v_1 > \frac{x_1 + L/2 - X_2(t)}{\Delta t},
\label{distveldeltatboundsimp}
\end{equation}
while the second and third coordinate of the velocity, $v_2$ and $v_3$, are sampled according to untruncated distributions ${\mathcal F}(v_j/\sigma_\mu)/\sigma_\mu$, for $j=2,3$. Using~(\ref{zetadef}) and symmetry, the $i$-th coordinates of velocities of particles entering the co-moving frame through the $i$-th left and right sides are sampled according to distributions proportional to
\begin{equation}
 \, {\mathcal F}\!\left(\frac{v_i-u_i}{\sigma_\mu} \right)
\qquad
\mbox{restricted to the subdomain}
\quad
v_i > \zeta \sigma_\mu 
+ V_{\mbox{\scriptsize f},i} \, ,
\label{distveldeltatboundsimp2}
\end{equation}
and
\begin{equation}
 \, {\mathcal F}\!\left(\frac{v_i-u_i}{\sigma_\mu} \right)
\qquad
\mbox{restricted to the subdomain}
\quad
v_i < - \zeta \sigma_\mu 
+ V_{\mbox{\scriptsize f},i} \, ,
\label{distveldeltatboundsimp3}
\end{equation}
respectively. An appropriate choice of the algorithm for sampling random numbers according to truncated distributions~(\ref{distveldeltatboundsimp2})--(\ref{distveldeltatboundsimp3}) will depend on the choice of ${\mathcal{F}}$, and it is provided for specific distributions in the next subsection.

\begin{table}
\caption{{\it Functions ${\mathcal G}$ and ${\mathcal K}$ defined by~$(\ref{functioncalgh})$ and parameters $a_1(\omega)$ and $a_2(\omega)$ satisfying the inequality~$(\ref{accrejcond})$ for some marginal density functions $\mathcal{F}$ introduced in Table~$\ref{tab:mdfxnscalinglist}$.} \label{table3}}
\centering
\vskip 3mm
\begin{tabular}{ccc}
\rowcolor{jobcolor}  & \color{white}{Gaussian Distribution} & \color{white}{Laplace Distribution} \\ 
$\mathcal{F}(q)$  & 
\raise -4.9mm \hbox{\rule{0pt}{11.6mm}}$\dfrac{1}{\sqrt{2\pi}}\,\exp\!\left(\! -\dfrac{q^2}{2}\right)$  & {\hskip -2mm}
\raise -4.9mm \hbox{\rule{0pt}{11.6mm}}$\dfrac{1}{2\sqrt{2}}\,\exp\!\left(\! -\dfrac{|q|}{\sqrt{2}}\right)$ 
\\
\hline
{\hskip -2mm} $\mathcal{G}(q)$ {\hskip -2mm}  
& \rule{0pt}{6mm} $\dfrac{1}{2} \, \mbox{erfc} \!\left( \dfrac{q}{\sqrt{2}}\right)$ {\hskip -2mm} 
& 
{\hskip -2mm} 
$\sqrt{2} \, {\mathcal F}(q) \, \mbox{sign}(q) + H(-q)$
{\hskip -2mm} 
\\
& &
\raise -9.5mm \hbox{\rule{0pt}{16mm}}$=
\begin{cases}
\dfrac{1}{2}\,\exp\!\left(\! -\dfrac{|q|}{\sqrt{2}}\right),
& \mbox{for} \; q \ge 0; \\
1 - \dfrac{1}{2}\,\exp\!\left(\! -\dfrac{|q|}{\sqrt{2}}\right),
& \mbox{for} \; q < 0.
\end{cases}
$
\\
\hline
$\mathcal{K}(q)$ &
${\mathcal F}(q) - q \, {\mathcal G}(q)$ 
 & \raise -3.5mm \hbox{\rule{0pt}{10.2mm}}$\displaystyle 2 {\mathcal F}(q) + \frac{|q| - q}{2}$ \\
\hline
$a_1(\omega)$ & \raise -9.4mm \hbox{\rule{0pt}{20.9mm}}$
\left\{
\begin{array}{ll}
\dfrac{{\mathcal G}(\omega)}{{\mathcal F}(\omega)} ,
&
\mbox{ for} \; \omega \ge 0;
\\
\dfrac{\sqrt{\pi}}{\sqrt{2}},
&
\mbox{ for} \; \omega \le 0. 
\\ 
\end{array}
\right.$
& $\sqrt{2}$ \\
\hline
$a_2(\omega)$ &\raise -10.4mm \hbox{\rule{0pt}{22.9mm}}$
\left\{
\begin{array}{ll}
\dfrac{1}{{\mathcal G}(\omega)},
&
\mbox{ for} \; \omega \ge 0;
\\
2 \, \exp
\!
\left(\!
\dfrac{\omega \sqrt{2}}{\sqrt{\pi}}
\right),
&
\mbox{ for} \; \omega \le 0. 
\\
\end{array}
\right.$
& $2
\,
\exp \! 
\left[
\dfrac{\omega}{\sqrt{2}}\right] 
$ \\
\hline
\end{tabular}
\end{table}

\subsection{Illustrative simulations for Gaussian and Laplace distributions}

\begin{figure}
(a) \hskip 8cm (b) \hfill\break
\epsfig{file=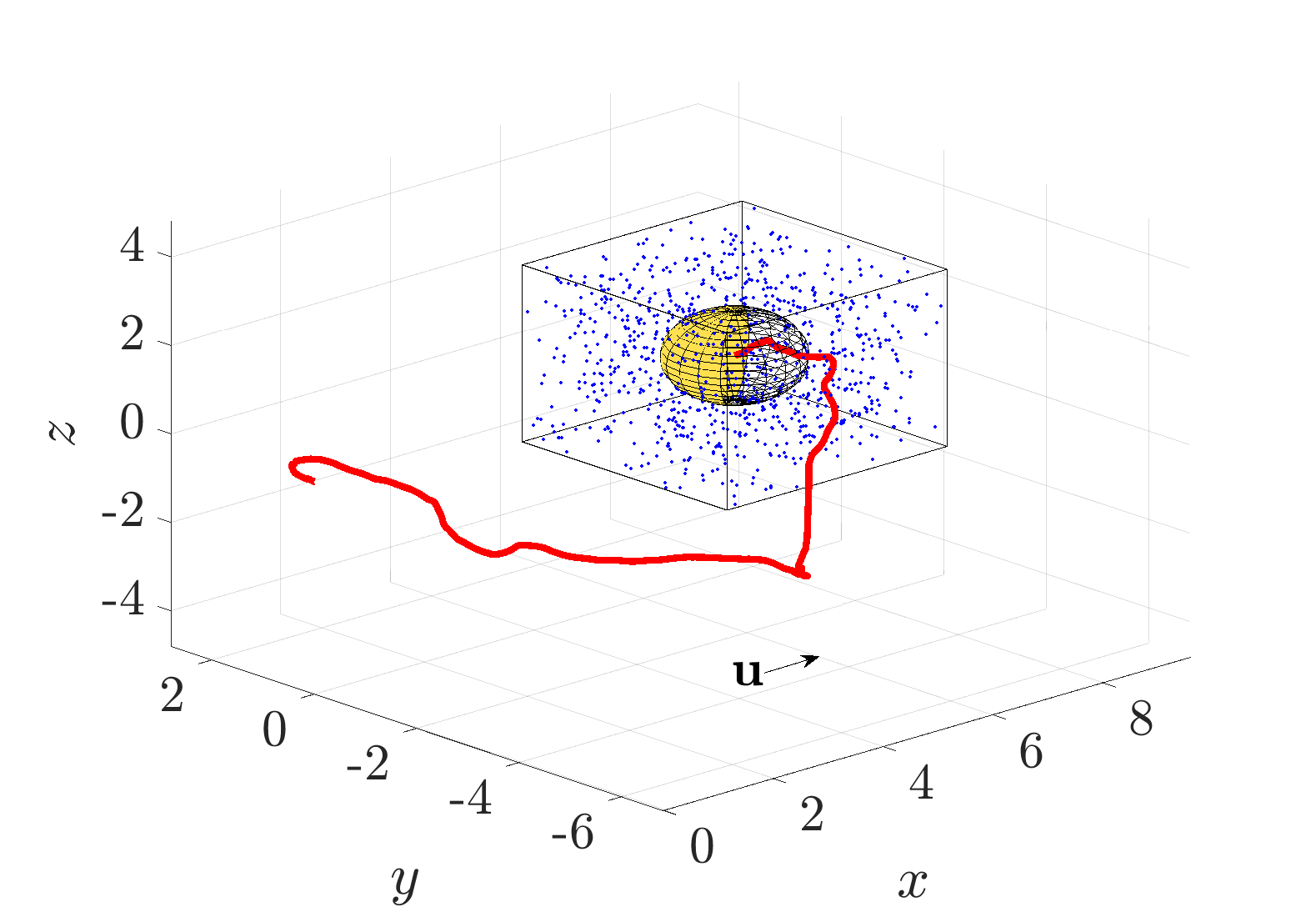,height=5.7cm} \hskip 4mm \epsfig{file=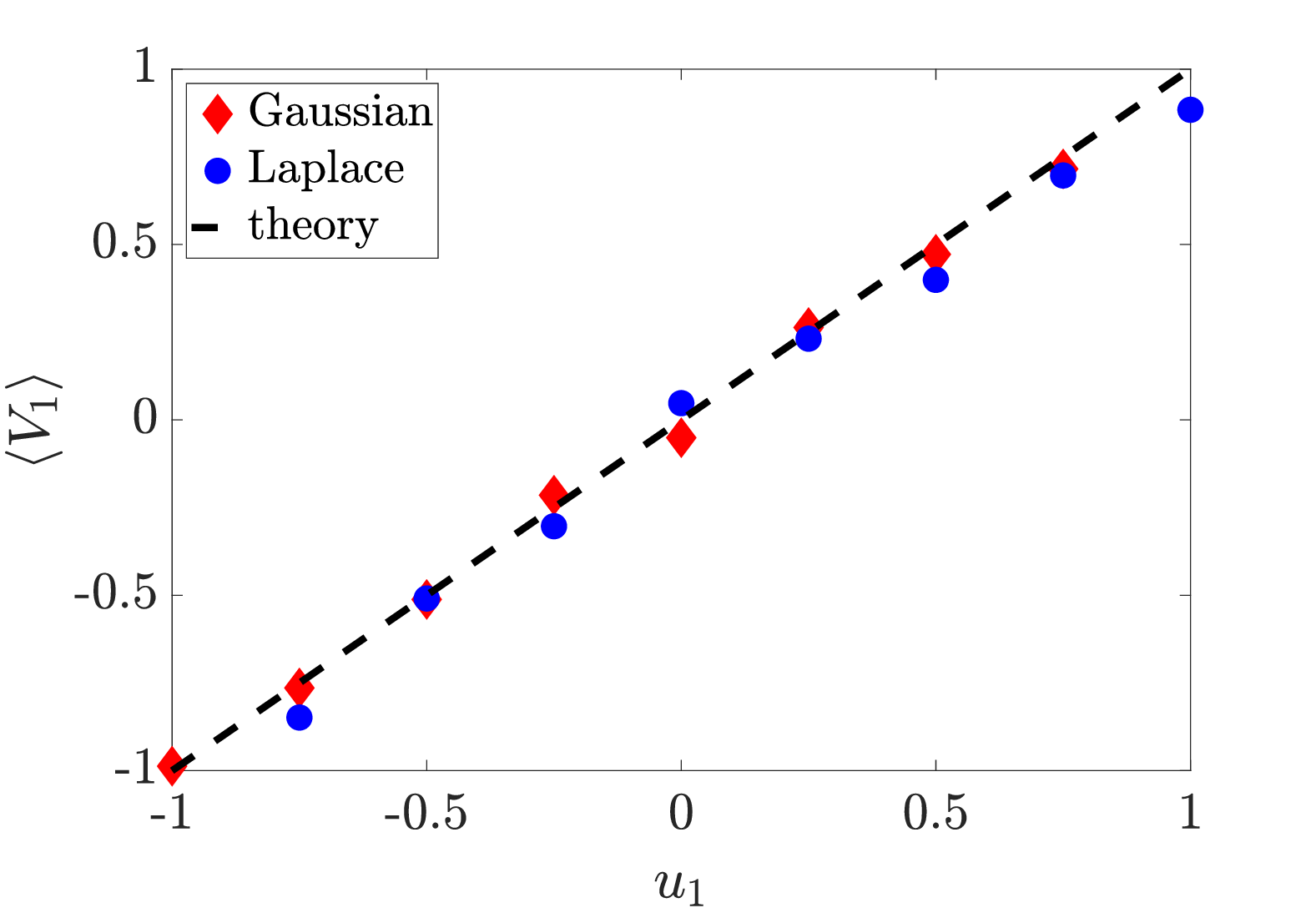,height=5.7cm}
\caption{\label{fig2} (a) {\it The initial trajectory of the large particle for $\mathbf{u}=(1,0,0)^{\mathrm{T}}$ in the time interval $t \in [0,10]$ is visualized as the red line together with the snapshot of positions of solvent particles at time $t=10$ (blue dots).} \hfill\break (b) 
{\it The first coordinate of the average velocity of the large particle, $\langle V_1 \rangle$, plotted as a function of the first coordinate of the fluid flow, $u_1$.}
}
\end{figure}

In Figure~\ref{fig2}(a), we present a trajectory of the solute particle calculated by Algorithm~[A1]-[A6] and visualized as the red line over a relatively short (dimensionless) time interval $t \in [0,10].$ In this simulation, we use the uniform flow ${\mathbf{u}}=(1,0,0)^{\mathrm{T}}$, dimensionless parameters
\begin{equation}
D = \gamma = m = R = 1,
\qquad
\Delta t = 10^{-5},
\qquad
M = 10^4,
\qquad
L = 4,
\label{parvalues}
\end{equation}
and the Gaussian distribution for ${\mathcal{F}}(q).$ In particular, we use the formulas for $a_1(\omega)$ and $a_2(\omega)$ for the Gaussian distribution given in Table~\ref{table3}. To justify these choices, we note that in the case of~$\omega \ge 0$, the formulas for $a_1(\omega)$ and $a_2(\omega)$ in Table~\ref{table3} imply that condition~(\ref{accrejcond}) can be rewritten as 
$$
\frac{{\mathcal G}(\zeta+\omega)}{{\mathcal G}(\omega)}
\,
\exp \! 
\left[
\frac{\zeta \, {\mathcal F}(\omega)}{{\mathcal G}(\omega)}
\right] 
\,\le\,
1
\qquad \mbox{for all} \quad \zeta > 0.
$$
The left-hand side is equal to $1$ for $\zeta=0$, which is the maximum value of the left-hand side for $\zeta \ge 0$, meaning that the condition~(\ref{accrejcond}) is satisfied. On the other hand, if $\omega<0$, then the formulas for $a_1(\omega)$, $a_2(\omega)$ and ${\mathcal G}(\omega)$ in Table~\ref{table3} for the Gaussian distribution imply that  the condition~(\ref{accrejcond}) can be rewritten as 
$$
\mbox{erfc} \!\left( \frac{\zeta+\omega}{\sqrt{2}}\right)
\,
\exp \! 
\left[
\frac{(\zeta+\omega) \sqrt{2}}{\sqrt{\pi}}
\right] 
\,\le\,
1
\qquad \mbox{for all} \quad \zeta > 0,
$$
where the left-hand side is equal to 1 at point $\zeta = -\omega$ and this is the maximum value of the left-hand side for $\zeta \ge 0$, meaning that the condition~(\ref{accrejcond}) is again satisfied. 

Using different formulas of $a_1(\omega)$ and $a_2(\omega)$ for $\omega \ge 0$ and $\omega < 0$ improves the acceptance probability~(\ref{accprobS4t}) of Algorithm [S1]-[S4], which can be further improved if we use tabulated functions, see~\cite{gunaratne2019multi} for further discussion. The initial velocity of the introduced solvent particle is sampled according to the truncated Gaussian distribution~(\ref{distveldeltatboundsimp}) using an acceptance-rejection method presented in the literature~\cite{Robert:1995:STN}.

Using our parameter values~(\ref{parvalues}), equations~(\ref{lambda3Dexp}) and~(\ref{sigmaold}) give
\begin{equation}
\mu = 10^4,
\qquad
\lambda_\mu \approx 14.96,
\qquad
\sigma_\mu \approx 100.
\label{derivedparameters}
\end{equation}
In particular, the volume of the co-moving frame available to solvent particles is $L^3-(4/3) \pi R^3$ and it contains around $\lambda_\mu (L^3-(4/3) \pi R^3) \approx 895$ solvent particles on average. The positions of solvent particles in the co-moving frame at time $t=10$ are visualized in Figure~\ref{fig2}(a) as blue dots.

Figure~\ref{fig2}(a) illustrates a relatively short simulation for time $t \in [0,10]$. Next, we increase the simulated time interval to $t \in \big[0, \, 2\!\times\! 10^3\big]$ using the same values of parameters~(\ref{parvalues}) and implementing Algorithm~[A1]-[A6] over $2 \!\times\! 10^8$ simulated time steps. We use uniform flow 
\begin{equation}
{\mathbf{u}}=(u_1,0,0)^{\mathrm{T}},
\qquad
\mbox{where} \quad
u_1 \in \left\{ 
-1, \, -\frac{3}{4}, \, -\frac{1}{2}, \, -\frac{1}{4}, \, 0, \, \frac{1}{4}, \, \frac{1}{2}, \, \frac{3}{4}, \, 1 \right\},
\label{uniformflowu1}
\end{equation} 
presenting the results of these nine long-time simulations in Figure~\ref{fig2}(b). In each simulation, we average the first component of the velocity of the solute particle, $\langle V_1 \rangle$, at every two time steps, so that we average over $10^8$ individual data points to calculate $\langle V_1 \rangle$. We plot this quantity in Figure~\ref{fig2}(b) as a function of $u_1$, confirming the theoretical result $\langle V_1 \rangle = u_1$.

Our calculated results for the Laplace distribution are also presented in Figure~\ref{fig2}(b), where we use the same parameter values~(\ref{parvalues})--(\ref{derivedparameters}) and the same nine underlying homogeneous flows~(\ref{uniformflowu1}) as in the Gaussian case. Using the formulas for ${\mathcal F}$, ${\mathcal G}$, ${\mathcal K}$, $a_1(\omega)$ and $a_2(\omega)$ in Table~\ref{table3} for the Laplace distribution and equation (\ref{accprobS4t}), we conclude 
that the algorithm [S1]-[S4] finishes in one iteration of Steps~[S2]-[S4] with probability
$$
\frac{a_2(\omega)\, {\mathcal K}(\omega)}{a_1(\omega)}
=
\left\{
\begin{array}{ll}
1 \, ,
&
\mbox{ for} ~~ \omega \ge 0; 
\\
\exp \! 
\left[
\sqrt{2}
\, \omega
\right]
+
\sqrt{2}
\, |\omega| \,
\exp \! 
\left[
\dfrac{\omega}{\sqrt{2}}
\right] 
\, ,
&
\mbox{ for} ~~ \omega \le 0. 
\\
\end{array}
\right.
$$
In particular, if $\omega \ge 0$, then the distribution~(\ref{shiftcalG}) is an exponential distribution and our choices of $a_1(\omega)$ and $a_2(\omega)$ given in Table~\ref{table3} ensure that every exponentially distributed random number calculated in Step~[S3] is accepted in Step~[S4], so that the algorithm in Table~\ref{tableaccrej} finishes in one iteration of steps~[S2]-[S4]. Initial velocities of the introduced solvent particles are sampled according to the truncated Laplace distribution~(\ref{distveldeltatboundsimp}). To do this, we can sample a random number according to the Laplace distribution and accept it, if it is inside the desired range of values. Such an acceptance-rejection algorithm will have its acceptance probability greater than 1/2, provided that the truncated Laplace distribution~(\ref{distveldeltatboundsimp}) includes both positive and negative values. If not, then the truncated Laplace distribution becomes an exponential distribution and we do not need to use an acceptance-rejection algorithm.

\begin{figure}
(a) \hskip 8cm (b) \hfill\break
\epsfig{file=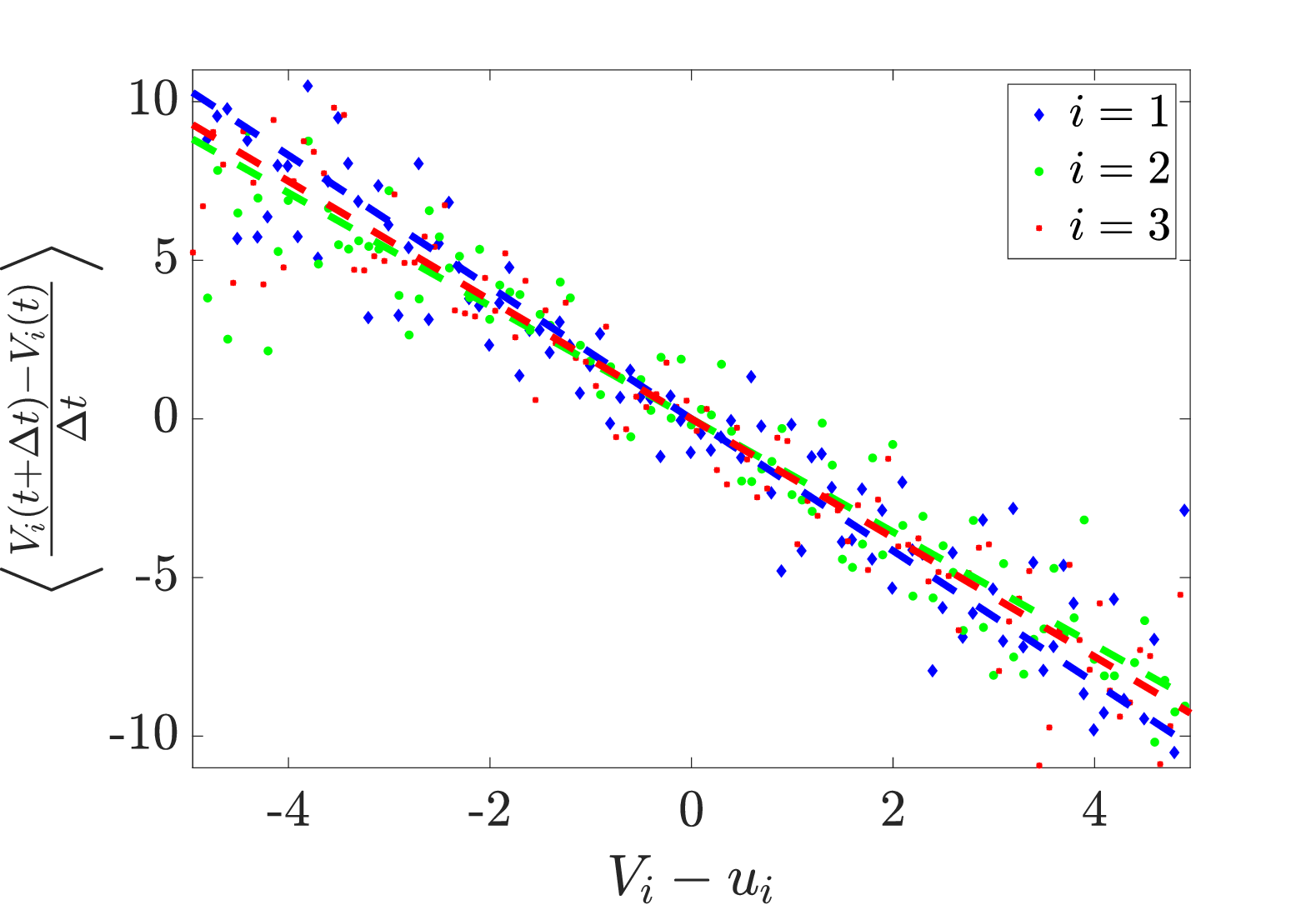,height=5.7cm} \hskip 4mm \epsfig{file=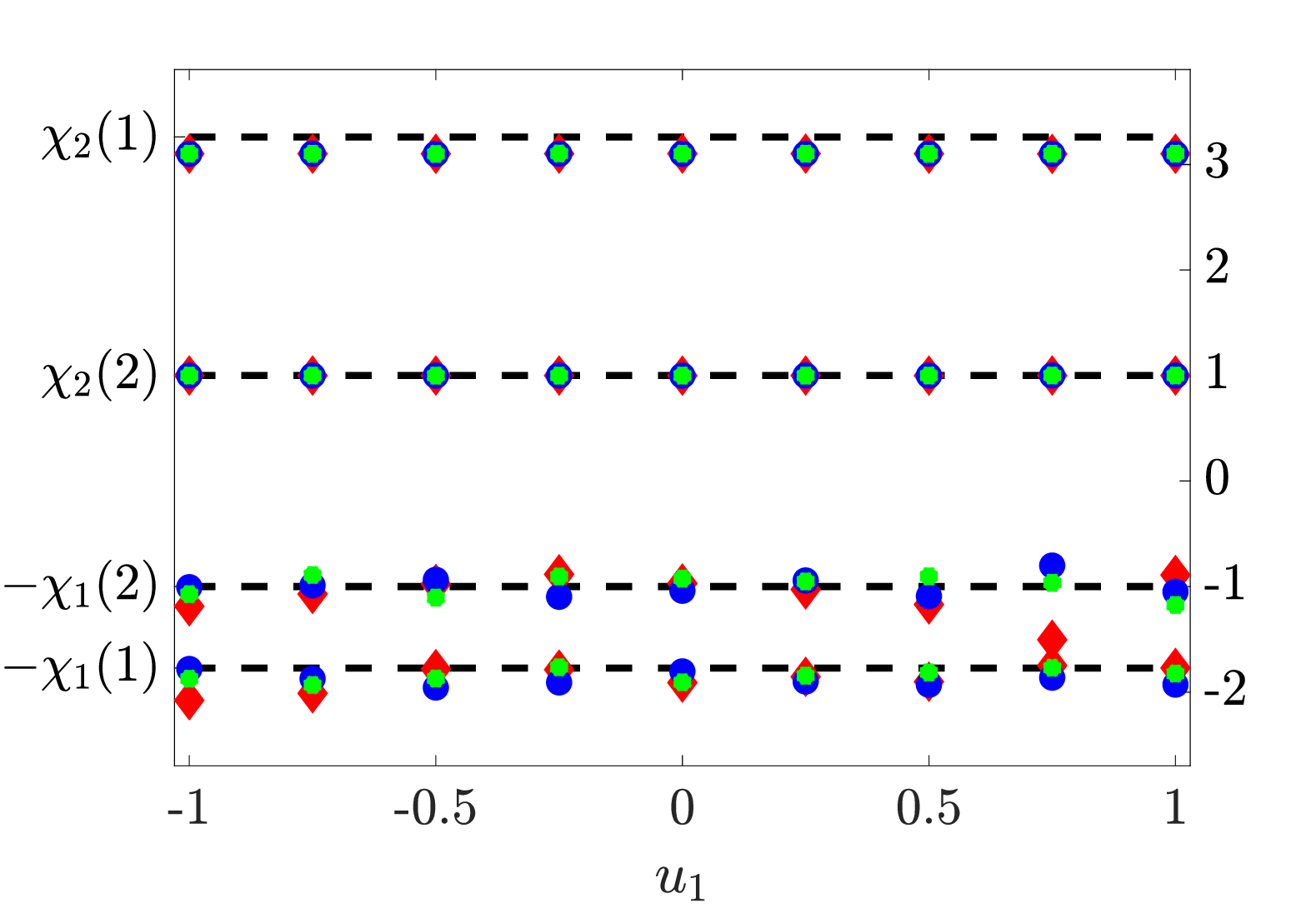,height=5.7cm}
\caption{\label{fig3} (a) 
{\it Estimation of the coefficient $\chi_1(1)$ defined in $\eqref{chiscalefactors}$ for the simulation of the large particle in flow $\mathbf{u}=(1,0,0)^{\mathrm{T}}$. We plot $\langle V_i(t+\Delta t) - V_i(t) \rangle/\Delta t$ as a function of $V_i - u_i$, for each coordinate $i=1,2,3$, and use the best linear fit to estimate the slope, $-\chi_1(1)$, shown in the second panel.} \hfill\break (b)  
{\it Coefficients~$\chi_1(\theta)$ and $\chi_2(\theta)$ estimated using~$(\ref{estdiff})$ from the simulation of the Algorithm~{\rm [A1]-[A6]} for $\theta=1$ (Laplace distribution) and $\theta=2$ (Gaussian distribution) as functions of flow velocity $u_1$.}}
\end{figure}

Figure~\ref{fig3}(b) shows the estimated values of coefficients~$\chi_1(\theta)$ and $\chi_2(\theta)$
for both Laplace ($\theta=1$) and Gaussian ($\theta=2$) distributions. To estimate $\chi_2(\theta)$, we calculate the average
\begin{equation}    
\chi_2(\theta)
\approx
\left(
\frac{1}{2 \, D \, \gamma^2} \left\langle \frac{(V_i(t+\Delta t) - V_i(t))^2}{\Delta t}
\right\rangle
\right)^{\!1/2}
\label{estdiff}
\end{equation}
during the same long-time simulations of the Algorithm~{\rm [A1]-[A6]}, which we have used for Figure~\ref{fig2}. To estimate $\chi_1(\theta)$, we plot $\langle V_i(t+\Delta t) - V_i(t) \rangle/\Delta t$ as a function of $V_i - u_i$ in Figure~\ref{fig3}(a) for the simulation with $u_1=1.$ Using the best linear fit for each coordinate $i=1,2,3$, we obtain three values of slopes; the resulting estimates of $-\chi_1(\theta)$ are shown in Figure~\ref{fig3}(b). All simulation results fall within the $\mathcal{O}(\sigma_\mu^{-1})$ error predicted by the theory.

\section{Discussion}\label{sec5}
We have derived a Langevin-type macroscopic description of a solute particle immersed in a heat bath of light point particles which are subject to a stationary flow. A range of flows and velocity distributions have been considered in Section~\ref{sec3} and efficient methods for simulating the microscopic system in a co-moving frame have been designed in Section~\ref{sec4}. Our results extend the theory for Brownian motion of a heavy particle to the case where the light heat bath particles follow a prescribed flow. We highlight some key points from our theoretical findings.

The result in Theorem \ref{theorem1} provides a general approach for determining the drift and diffusion terms of the Langevin dynamics~(\ref{SDEmainX})--(\ref{SDEmainY}) governing the motion of the heavy particle. Many distribution functions -- including generalized Gaussian distributions -- obey the symmetry property manifesting as the asymptotic scaling \eqref{mdfxnScale}. For heat bath particles obeying such velocity distributions, the approach taken in Theorem \ref{theorem2} allows us to simplify the expressions in Theorem \ref{theorem1} greatly, resulting in explicit formulas (up to an error $\mathcal{O}(\sigma_\mu^{-1})$, which is small in the limit where the heat bath particle masses are much smaller than that of the heavy particle) for the drift and diffusion coefficients in the generalized Ornstein-Uhlenbeck process~(\ref{SDEmainY}). The influence of the particular velocity distribution chosen for the heat bath particles manifests through factors which scale the drift and diffusion terms. When the heat bath particle velocities obey a Gaussian distribution (Maxwell-Boltzmann statistics) these scale factors reduce to unity, simplifying to scales common in the literature. When the heat bath particle velocity distribution has a heavy tail, we find that the scale factors are greater than unity (so, faster heat bath particles are more common than in the Gaussian case), while when the heat bath particle velocity distribution is thin-tailed the scale factors are less than unity (hence, slower heat bath particles are more common than in the Gaussian case). These theoretical findings are verified in our numerical simulations in Section~\ref{sec4}. 

In addition to the velocity distribution-dependent scale factors discussed above, the drift term in the obtained generalized Ornstein-Uhlenbeck process also depends upon the flow. Prior literature (see, for instance, \cite{rubi1988brownian, gotoh1990brownian, katayama1996brownian, garbaczewski1998diffusion, orihara2011brownian, wang2022generalized}) has generally assumed that the drift term $\boldsymbol{\alpha}$ scales like $\boldsymbol{\alpha} \sim -(\mathbf{V-\mathbf{u}(\mathbf{X})})$, with a formal derivation for certain linear shear flows appearing more recently in~\cite{dobson2013derivation}. This approximation is reasonable for heavy particles of negligible size or for linear flows. However, for heavy particles of finite size in nonlinear flows, this scaling is not complete, and we show via Theorem~\ref{theorem2} that there are finite-size effects when the heavy particle is large enough to interact with the geometry of the flow. Therefore, the appropriate leading-order scaling for the drift term is $\boldsymbol{\alpha} \sim -(\mathbf{V-\mathbf{u}(\mathbf{X})} - \mathbf{A}_R)$ where the correction term $\mathbf{A}_R$ is defined by equation~\eqref{A_correction} and depends upon both the flow geometry and the size $R$ of the heavy particle. If the underlying flow of heat bath particles is smooth enough over the problem domain $\Omega$, Theorem~\ref{theoremGeometry} implies $\mathbf{A}_R=\mathcal{O}( R^2)$ in the particle size $R$. Therefore, when $\mathcal{O}( R^2)$ contributions are larger than $\mathcal{O}(\sigma_\mu^{-1})$ contributions, it is necessary to include the finite-size correction to the drift term. We have demonstrated how this correction term manifests in the specific examples of Poiseuille and boundary layer flows.

There are a number of ways our work might be extended. 
Regarding the heavy particle, we have assumed a perfect sphere and neglected rotation or angular momentum. Including angular in addition to linear momentum would greatly complicate the governing equations, yet would permit a more realistic view of how a finite-size heavy particle moves within a flow. Inclusion of solid-body motions would also allow for the consideration of non-spherical particles. Furthermore, we have considered one heavy particle immersed within a heat bath of many light particles, and the extension to two or some finite number of heavy particles would be another possible generalization. Collisions with the heat bath particles bias the motion of heavy particles toward the imposed flow, and hence the motion of heavy particles is expected to be along the mean flow yet may differ from this mean flow greatly after interacting with other heavy particles either directly by collisions~\cite{Franz} or indirectly by hydrodynamic interactions through the solvent~\cite{Ermak:1978:BDH,Rolls:2018:MPB}. Regarding the nature of the heat bath flows considered, although we have considered stationary flows for the heat bath particles it should be possible to account for time-varying flows~\cite{batchelor1967introduction,landau2013fluid}. Likewise, while we consider incompressible flows, the extension to compressible flows with a time-varying number of light particles within a given volume element would be of interest for applications in gas dynamics.

\vskip 10mm

\noindent
{\bf Acknowledgments.} RVG was a Visiting Scholar at Merton College and the Mathematical Institute, University of Oxford, UK, and a CNRS sponsored INP Guest Scientist at the Institut de Physique de Nice, France in the first half of 2023 while parts of this paper were completed, and thanks these organizations for their hospitality and support over this time. This work was supported by the Engineering and Physical Sciences Research Council, grant
number EP/V047469/1, awarded to RE.

\end{document}